\documentclass[a4paper,10pt]{article}

\usepackage[pages=all, color=black, position={current page.south}, placement=bottom, scale=1, opacity=1, vshift=5mm]{background}
\SetBgContents{
} 
\usepackage[margin=0.92in]{geometry}

\title{Moderate-length lifted quantum Tanner codes}
\author{Virgile Guemard$^{1, 2}$ \and Gilles Zemor$^{3,4}$}

\date{
	$^1$Aix Marseille Université, I2M, UMR 7373, 13453 Marseille, France\\%
        $^2$Inria Paris, France \\
        $^3$Institut de Mathématiques de Bordeaux, UMR 5251\\
        $^4$Institut universitaire de France\\[2ex]%
	\today
}
 
\usepackage{amsmath}
\usepackage{amsthm}
\usepackage{amssymb}
\usepackage{soul}

\usepackage[utf8]{inputenc}
\usepackage{hyperref}
\hypersetup{
	unicode,
	pdfauthor={Virgile Guémard, Gilles Zémor},
	pdftitle={Moderate-length lifted quantum Tanner codes},
	pdfsubject={Quantum Tanner codes},
	pdfkeywords={article},
	pdfproducer={LaTeX},
	pdfcreator={pdflatex}
}

\usepackage[sort&compress,numbers,square]{natbib}

\theoremstyle{plain}
\newtheorem{theorem}{Theorem}[section]

\newtheorem{lemma}[theorem]{Lemma}

\newtheorem{proposition}[theorem]{Proposition}

\theoremstyle{definition}
\newtheorem{definition}[theorem]{Definition}
\newtheorem{example}[theorem]{Example}
\newtheorem{remark}[theorem]{Remark}

\usepackage{graphicx, color}
\graphicspath{{fig/}}

\usepackage{algorithm, algpseudocode} 
\usepackage{mathrsfs}

\usepackage{dcolumn}% Align table columns on decimal point
\usepackage{bm}% bold math
\usepackage[english]{babel}
\usepackage[OT1]{fontenc}

\usepackage[colorinlistoftodos, color=green!40, prependcaption]{todonotes}

\usepackage{amsthm}
\usepackage{mathtools}
\usepackage{physics}
\usepackage{xcolor}
\usepackage{graphicx}
\usepackage{tikz-cd}
\usepackage{bbold}

\usepackage{adjustbox}
\usepackage{placeins}

\usepackage{csquotes}
\usepackage{booktabs}
\usepackage[caption=false]{subfig}
 \makeatletter
\renewcommand{\fnum@figure}{FIG. \thefigure}
\makeatother
\usepackage{tikz}
\usetikzlibrary{decorations.markings}

\usepackage{amsmath,mleftright}
\usepackage{xparse}

\NewDocumentCommand{\evalat}{sO{\big}mm}{%
  \IfBooleanTF{#1}
   {\mleft. #3 \mright|_{#4}}
   {#3#2|_{#4}}%
}

\usepackage{amsthm}

\usepackage{extpfeil}
\usepackage{tikz-cd}
\usepackage{svg}

\usepackage{titlesec}
\titleformat{\section}%[runin]
  {\normalfont\bfseries}{\thesection}{1em}{}
\titleformat{\subsection}%[runin]
  {\normalfont\bf}{\thesubsection}{1em}{}{\normalfont\bfseries}
\titleformat{\subsubsection}%[runin]
  {\normalfont\bfseries}{\thesubsubsection}{1em}{}
\renewcommand{\thesubsection}{\thesection.\arabic{subsection}}
\renewcommand{\thesection}{\arabic{section}}
\renewcommand{\thesubsubsection}{\thesubsection.\arabic{subsubsection}} 

\usetikzlibrary{arrows}

\usepackage{leftindex}
\makeatletter
\newtheorem*{rep@theorem}{\rep@title}
\newcommand{\newreptheorem}[2]{%
\newenvironment{rep#1}[1]{%
 \def\rep@title{#2 \ref{##1}}%
 \begin{rep@theorem}}%
 {\end{rep@theorem}}}
\makeatother
\newreptheorem{proposition}{Proposition}

\begin{document}
	\maketitle

\begin{abstract}
We introduce new families of quantum Tanner codes, a class of quantum codes that first appeared in the work of Leverrier and Z\'emor \cite{Leverrier2022}. These codes are built from two classical Tanner codes, for which the underlying graphs are extracted from coverings of 2D geometrical complexes, and the local linear codes are tensor-products of cyclic or double-circulant linear codes. The advantage of code lifting is that, for any lift of odd index $t$ of an $[[n,k,d]]$-code, we can adapt the study of the transfer homomorphism arising in cellular homology to describe symmetries of its logical operators and to establish that its dimension is lower bounded by $k$, and its distance is upper bounded by $t\cdot d$. Moreover, when the dimension of the lifted code is equal to $k$, its distance is lower bounded by $d$. These parameter bounds also apply to the previous methods of code lifting \cite{GuemardLiftIEEE}. Finally, We present several explicit families, and identify instances of moderate length quantum codes which are degenerate, have low check weight, and whose distance surpasses the square root of the code length. Among them, we report the existence of a $[[96,2,12]]$-code whose distance growth saturates our bound, and for which half of the checks are of weight 8 and the other half of weight 4.
\end{abstract}

\section{Introduction}

Quantum error correcting codes are believed to be unavoidable constituents of quantum computing architectures. Among them, quantum low-density parity-check (LDPC) codes are the leading candidates for physical implementations, because their low connectivity translates into a limitation of the error propagation between the hardware constituents. Significant progress has been made recently in the construction of good quantum LDPC codes \cite{Panteleev2021,Leverrier2022,LinHsie2022}, i.e. codes families with constant rate and relative distance in the asymptotic limit. These families can constitute a valuable source of inspiration to devise new types of codes of moderate length, following the effort in the design of near term implementable codes \cite{Panteleev2020,Pryadko2022,Bravyi2024}. \par

In 2022, Leverrier and Z\'emor \cite{Leverrier2022} generalized the Sipser-Spielman version \cite{Sipser1996} of classical Tanner codes \cite{Tanner1981ARA} to devise a new family of asymptotically good CSS codes, which they call quantum Tanner codes. The methodology of \cite{Leverrier2022} relies on two components: a family of geometric complexes and a set of short linear codes. They instantiated it by considering a class of 2D square complexes known as left-right Cvayley complexes, first used in the coding context in \cite{Dinur2021}. Given one of these complexes, its face-vertex incidence structure generates two graphs, each used to define a classical Tanner code. The local structure of the complex allows for classical tensor-product codes to be used as local codes, which is key to ensure commutation between stabilizers of the quantum code. Although this CSS construction is state-of-the-art in the asymptotic regime, short-length examples based on the same ideas have not yet been developed.\par

This article introduces new families of quantum Tanner codes, inspired by the Leverrier-Z\'emor construction \cite{Leverrier2022,leverrier2022decoding}. The main novelty of our approach is that we start by defining a quantum Tanner code on a square complex, which is not necessarily a left-right Cayley complex or a constant-degree graph as in \cite{Mostad2024}. Recall that exploring beyond constant degree Tanner graphs has been key to devise capacity-approaching classical LDPC codes \cite{910578}. We then use our construction to build larger quantum Tanner codes based on the idea of code lifting.\par

Our approach to lift is also largely inspired by the method of \cite{Panteleev2021} to lift classical Tanner codes. Here, instead of using only a graph and a local code, we work with a 2-dimensional square complex over which two classical Tanner codes can be defined, forming an input CSS code. A covering of the square complex then allows us to simultaneously define lifts of these two classical Tanner codes. The maximum weight of the rows and columns of the lifted check matrices remains equal to that of the input code, making it well-suited for LDPC constructions. A peculiarity of our lifting method is that the size of the overlapping support of an $X$- and a $Z$-check can be reduced under lifting. This shows that this method is different from the one defined in \cite{GuemardLiftIEEE} for which the size of the overlap is conserved under lifting.\par
The drawback of our method is that lifting a CSS code does not always yield a valid CSS code. To ensure that lifting yields a valid CSS code, the symmetries of the local codes must be compatible with the structure of the underlying 2D complex. The advantage is that for lifts of odd indices, we can adapt the study of the transfer homomorphism, arising in cellular homology \cite{HatcherTopo}, to establish a lower bound on the dimension of lifted codes and to describe the structure of their logical operators. Namely, a lifted $[[\tilde n, \tilde k,\tilde d]]$-code of odd index $t$ of a $[[n,k,d]]$-code, has $\tilde n=tn$, $\tilde k\geq k$ and $\tilde d\leq td$. Moreover, when $\tilde k=k$, it follows that $\tilde d\geq d$.  The technics and parameter bounds developed here also apply to the code lifting method of \cite{GuemardLiftIEEE}, based on coverings of the Tanner cone-complex.\par

We illustrate the lifting of quantum Tanner codes through explicit examples. In \cite{GuemardLiftIEEE}, new code constructions were introduced using certain 2D geometric complexes, namely square subdivisions of presentation complexes of finitely presented groups; the CSS codes were built by interpreting the vertices, in the subdivision of these complexes, as parity checks and qubits. Here, following \cite{Leverrier2022}, we adopt a different perspective and interpret the faces as qubits. The novelty of our method is to first construct short-length quantum Tanner codes on these subdivided complexes by fixing a specific configuration of the local codes. We then obtain the lifted Tanner code by generating coverings of these square complexes and lift the local code configuration to the covering spaces. In certain cases, the effect of the local code is to suppress the short codewords, causing the support of minimal codewords to spread out across the complex. In this way, we are able to construct LDPC codes whose distance squared exceeds the blocklength.\par

It has proven challenging to determine analytically the parameters of our lifted Tanner codes. We therefore follow a numerical approach. For selected input quantum Tanner codes, we generate all possible lifts of index 1 to 30, producing codes of length less than 800 qubits. An upper bound for the distance is computed using GAP and the package QDistRnd \cite{Pryadko2022}. We report only the codes with the highest parameters and leave a more rigorous analysis for future work.\\

This article is organized as follows.  In Section \ref{section Preliminaries}, we set the notation for graphs, square-complexes and covering maps. We also review classical Tanner codes and our preferred method for lifting them.\par

In Section~\ref{section: quantum Tanner codes and their lifts}, we present a method to lift a quantum Tanner code. The idea is a generalization of the lift of a classical Tanner code, where the central object is a covering of a square complex rather than a graph. We also extend the study of the transfer homomorphism from cellular homology to characterize the structure of any lifted codes of odd index.\par

Lastly, in Section~\ref{section:New constructions of quantum Tanner codes}, we introduce new explicit constructions of quantum Tanner codes. We then perform an exhaustive numerical search for lifted codes and report specific instances with the highest distances.

\section{Preliminaries}\label{section Preliminaries}

\subsection{Graphs and square-complexes}\label{Graphs, square-complexes and covering maps}

In this work, crucial components of the linear and quantum code constructions are graphs and square complexes. We therefore set some notations and state elementary facts about these objects, which will be used throughout this article.\par

\paragraph{Graphs.} We denote by $\mathcal G:=(V,E)$ an undirected graph with set of vertices $V$ and edges $E$. An \textit{edge} between two vertices $u$ and $v$ is an unordered pair $\{u,v\}=\{v,u\}$. In this work, graphs have no self-loops, i.e., edges of the form $\{u,u\}$, unless otherwise stated, but they can have multi-edges. A \textit{directed edge} associated to $\{u,v\}\in E$ is an edge for which the order matters, written $e=[u,v]$ for an edge with initial vertex $u$ and terminal vertex $v$. The inverse of an edge is defined as $[u,v]^{-1}:=[v,u]$. We denote by $E^*$ the set of oriented edges. Its cardinality is twice that of $E$. A path of length $t$ in the graph is a sequence of directed edges $e_1\cdot e_2\dots e_t$ such that the terminal vertex of $e_i$ is the initial vertex of $e_{i+1}$, for $i<t$. A path is closed if the terminal vertex of $e_t$ is the initial vertex of $e_1$. The edge-neighborhood of a vertex $v$, denoted $E(v)$, is the set of edges connected to $v$, and its cardinality is the \textit{degree} of $v$. A \textit{$t$-cycle} in $\mathcal G$ is a connected subgraph constituted of $t$ vertices (of degree 2 in the subgraph) and $t$ edges, forming a single cycle in $\mathcal G$. \\

\paragraph{Square complexes.} A square complex is defined as a graph with a set of preferred subgraphs which are \textit{4-cycles}. We denote it $\mathcal S:=(V,E,F)$, where $(V,E)$ is the underlying graph, also called its 1-skeleton, and $F$ is a set of 4-cycles. Notice that a square complex with an empty set of faces is also a graph. We say that $\mathcal S$ is a \textit{bipartite square complex} if $(V,E)$ is a bipartite graph. A face $f\in F$ can be described by its corresponding $4$-cycle, or by a closed path, such as $\alpha_f=[v_1,v_2].[v_2,v_3].[v_3,v_4].[v_4,v_1]$. The \textit{face-neighborhood} of a vertex $v$ is the set of faces to which it is incident, and is denoted $F(v)$. Moreover, we say that two faces are adjacent in $\mathcal S$ if they share a common edge.\par

\paragraph{Topological realization.}  the \textit{topological realization} of a square complex is a topological space divided into 2D and 1D subspaces. For us, it is most convenient to view the topological realization of a square complex as a \textit{regular }2D cell complex \cite{HatcherTopo} (or CW complex). These are topological spaces obtained by successively gluing cells, with gluing maps being homeomorphisms\footnote{\hangindent=1em  \hangafter=1 Examples of gluing maps that are not homeomorphisms: the map $\partial \mathbb{I} \to p$, mapping both end-points of an interval to the same point, yielding a 1D-sphere; the map $\partial D^2\to p$, mapping the boundary of the disk, a circle, to a point, yielding a 2D-sphere.}. More precisely, in the topological realization of a square complex, for each combinatorial face, we glue an actual face, which is homeomorphic to a two-dimensional disk $D^2$, along the path formed by a 4-cycle. The resulting complex is then endowed with the topology of a CW complex.\par
Here, when we construct graphs and square complexes, we identify them with their topological realization. This is because we consider their \textit{fundamental group}, defined in the next section, and perform geometrical operations on them. It is important to note that every abstract 2D complex defines a topological one and vice versa.\par

\subsection{Fundamental group}\label{section fundamental group}

For considerations related to covering spaces in Section \ref{section covering maps}, it will be important to introduce the fundamental group, a central object in algebraic topology \cite{HatcherTopo, Lyndon2001}.\par

\paragraph{Definition.} We first recall that given a topological space $K$ the \textit{fundamental group} $\pi_1(K,v)$ is the group of homotopy classes of loops based at $v$, with group operation given by loop concatenation. For any other choice of basepoint $v'$, the groups $\pi_1(K,v)$ and $\pi_1(K,v')$ are isomorphic. In what follows, we will often omit the choice of basepoint and simply write $\pi_1(K)$. A space is called \textit{simply connected} if its fundamental group is trivial.\par

\paragraph{Combinatorial approach for square-complexes.} As previously mentioned, a square complex can be viewed as a 2D cell complex via its topological realization. In that setting, the fundamental group is defined as above. Moreover, there is an equivalent description of the fundamental group that solely relies on the combinatorial data of the square complex. This description, which we now review, is convenient for computations, and it is often used to determine presentations of the fundamental group.\par

Given a closed path $e_1\cdot e_2\dots e_n $ an elementary reduction consists of either removing two consecutive edges $e_i$ and $e_{i+1}$ if $e_{i+1} = e_i^{-1}$, or removing four consecutive edges if the corresponding sequence forms a face in the square complex $K$. The equivalence relation generated by such reductions is called \textit{homotopy}; it is compatible with concatenation of closed path, where the concatenation of two paths $\alpha:=e_1\cdot e_2\dots e_t$ and $\alpha':=e'_1\cdot e'_2\dots e'_{t'}$, such that the terminal vertex of $e_t$ is the initial vertex of $e_1'$, is the path  $\alpha\cdot\alpha':=e_1\cdot e_2\dots e_t\cdot e'_1\cdot e'_2\dots e'_{t'}$. The fundamental group $\pi_1(\mathcal S,v)$ is then isomorphic to the set of homotopy classes of closed paths based at $v$, with concatenation as the group operation. It can be shown that $\pi_1(\mathcal S,v) \cong \pi_1(\mathcal S,v')$ for any two vertices $v, v'$, if its 1-skeleton $\mathcal S^1:=(V,E)$ is a connected graph, and thus we often write $\pi_1(\mathcal S)$.\par

If a square complex $(V,E,F)$ has no faces, i.e. it is just the connected graph $(V,E)$, then the fundamental group of its topological realization is a free group of rank $|E| - |V| + 1$. For a general 2D complex having a graph $\mathcal G$ for 1-skeleton, the fundamental group is given by the quotient $\pi_1(\mathcal G)/N$, where $N$ is the normal closure of the subgroup generated by elements corresponding to 4-cycles.\par

\subsection{Covering maps}\label{section covering maps}

We now review some results about covering maps from \cite{GROSS1977273,HatcherTopo, Lyndon2001}. While some of them apply to topological spaces in general, we will only consider them in the context of graphs and square-complexes of finite degrees with a single connected component.\par

\paragraph{Graph coverings.} Let $\mathcal G$ be a graph, that we identify with its topological realization. A map between graphs $p:\mathcal G' \to \mathcal G$ is a covering if and only if it is continuous, it maps interior of edges homeomorphically onto interior of edges, and every vertex of $\mathcal G'$ has a neighborhood on which the restriction of $p$ is a homeomorphism.\par
For graphs, we may prefer a combinatorial definition of covering. Let us now see the graphs above simply as sets of vertices and edges with no self-loops. Then, a graph covering is a surjective map $p:\mathcal G' \to \mathcal G$ that maps edges to edges and vertices to vertices and such that the restriction of $p$ to the edge neighborhood of every vertex is a bijection. It follows that a vertex $v$ in $\mathcal G'$ and its image $p(v)$ have the same degree. The cardinality of the inverse image of any vertex or edge is a constant called the \textit{index} of the covering.\par
From  \cite{10.5555/29358}, such a map can be given a topological realization, which is a continuous function from a topological realization of $\mathcal G'$ to a topological realization of $\mathcal G$ whose restriction to the interior of any edge of $\mathcal G'$ is a homeomorphism and is consistent with $p$. A map $p$ between graphs is a combinatorial covering if and only if a topological realization of $p$ is a covering in the usual sense.\par

\paragraph{Permutation voltage.} We now summarize a construction for graph coverings that can generate all possible finite coverings, as shown in \cite{GROSS1977273}. Let $r>0$ be an integer and $S_r$ be the symmetric group on the set $[r]:=\{0,\dots,r-1\}$. Given a graph $\mathcal G=(V,E)$, we first assign a permutation $\pi_e\in S_r$ to each oriented edge $e\in  E^*$, with the constraint that $\pi_{e^{-1}}=\pi_{e}^{-1}$. This step is called a \textit{permutation voltage assignment}. We construct a graph $\tilde{\mathcal G}$, called a \textit{lifted graph}, with set of vertices in bijection with $V\times [r]$ and set of edges in bijection with $E\times [r]$, where an element $c$ of $\tilde{\mathcal G}$, vertex or edge, is written $(c,i)$ for $i\in [r]$. An edge $(e,i)$ in the graph $\tilde{\mathcal G}$, where $e=[u,v]$, connects $(u,i)$ and $(v,\pi_e (i ))$.\par

As proved in \cite{GROSS1977273}, given a lifted graph constructed as above, it is possible to define a \textit{covering map} $p:\tilde{\mathcal G}\rightarrow  \mathcal G$ acting on an element $(c,i)$, vertex or edge, by $p((c,i))=c$. Conversely, for any $r$-sheeted graph covering $p:\mathcal G'\rightarrow  \mathcal G$, there exists a permutation voltage assignment of $\mathcal G$ in $S_r$, such that the associated lifted graph is isomorphic to $\mathcal G'$. We indeed verify that all the conditions are verified, and in particular a vertex $(v,i)$ and its projection $v$ have the same degree. Moreover, an edge $\tilde e=[\tilde u,\tilde v]$ projects onto an edge $e=[p(\tilde u),p(\tilde v)]$.\par

\paragraph{Square complex coverings.} Given a square complex $\mathcal S=(V,E,F)$, an index-$r$ covering $\tilde{\mathcal S}$ is a graph covering $p:\tilde{\mathcal G}\rightarrow  \mathcal G$ which does not break the faces. That is, if $\alpha$ is a 4-cycle in $\mathcal G$ describing a face of $F$, then $p^{-1}(\alpha)$ is a set of $r$ disjoint 4-cycles. The set of faces in the covering is in bijection with $F\times [r]$, and given a face $f\in F$ described by the path 
\[
e_1.e_2.e_3.e_4=[v_1,v_2].[v_2,v_3].[v_3,v_4].[v_4,v_1],
\]
then for all $i\in [r]$, the face labeled $(f,i)$ in $\tilde{\mathcal S}$ can be described by the following 4-cycle:
\begin{center}
\begin{tikzpicture}[scale=1,
    every node/.style={scale=1},
    decoration={markings, mark=at position 0.5 with {\arrow{stealth}}}, 
    postaction={decorate}
  ]
  % Nodes
  \node[draw, circle, fill=black, inner sep=0.6pt, label=left:{$(v_1,i)$}] (v1) at (0,2) {};
  \node[draw, circle, fill=black, inner sep=0.6pt, label=right:{$(v_2,\pi_{e_1}(i))$}] (v2) at (3,2) {};
  \node[draw, circle, fill=black, inner sep=0.6pt, label=right:{$(v_3,\pi_{e_2}\circ \pi_{e_1}(i))$}] (v3) at (3,0) {};
  \node[draw, circle, fill=black, inner sep=0.6pt, label=left:{$(v_4,\pi_{e_3}\circ\pi_{e_2}\circ \pi_{e_1}(i))$}] (v4) at (0,0) {};

  % Edges with mid arrows
  \draw[postaction={decorate}] (v1) -- node[above] {$(e_1,i)$} (v2);
  \draw[postaction={decorate}] (v2) -- node[right] {$(e_2,\pi_{e_1}(i))$} (v3);
  \draw[postaction={decorate}] (v3) -- node[below] {$(e_3,\pi_{e_2}\circ \pi_{e_1}(i))$} (v4);
  \draw[postaction={decorate}] (v4) -- node[left] {$(e_4,\pi_{e_3}\circ\pi_{e_2}\circ \pi_{e_1}(i))$} (v1);
\end{tikzpicture}
\end{center}
in other words, for any $4$-cycle described by a closed path $e_1.e_2.e_3.e_4$ corresponding to an element of $F$, the permutation $\pi_{e_4}\circ\pi_{e_3}\circ\pi_{e_2}\circ \pi_{e_1}$ is the identity. We can then define an index-$r$ covering map $p:\tilde{\mathcal S}\to \mathcal S$ acting on an element $(c,i)$ (vertex, edge or face) by $p((c,i))=c$, and therefore $p$ maps faces to faces. The inverse image of any vertex, edge or face is a set of $r$ vertices, edges or faces.\par

\paragraph{Deck transformations.} We now consider a square-complex $\mathcal S$, possibly with an empty set of faces, i.e. a graph. Given a covering $p:\tilde{\mathcal S}\to \mathcal S$, a \textit{deck transformation}\footnote{This name suggests an analogy to the mixing of a deck of cards, which are called sheets in the context of covering spaces.} is an automorphism $\psi:\tilde{\mathcal S}\rightarrow \tilde{\mathcal S}$ such that $p\circ \psi=p$. The set of deck transformations, endowed with the operation of composition of maps, forms a group denoted $\operatorname{deck}(p)$. It is called the \textit{group of deck transformations} and acts on the left on $\tilde{\mathcal S}$. A \textit{Galois covering} is a covering enjoying a left, free and transitive action of $\operatorname{deck}(p)$ on the fiber. For any Galois covering, $p:\tilde{\mathcal S}\rightarrow \mathcal S$, it can be shown that $\operatorname{deck}(p)\setminus \tilde{\mathcal S}\cong \mathcal S$.\par 
In the rest of this section, by a covering $p:\tilde{\mathcal S}\to \mathcal S$ we always mean a \textit{finite} \textit{connected} covering, i.e. one whose index is finite and for which $\tilde{\mathcal S}$ is also a connected complex. The following crucial results of this section apply in this context. \par

\paragraph{Classification of coverings.} The theory of connected covering maps over a complex depends on the structure of its fundamental group in the following way. For every index $r$ subgroup $H $ of the fundamental group $\pi_1(\mathcal S)$ there exists an index $r$ covering $p:\mathcal S_H\to \mathcal S$, mapping the basepoint $\tilde v$ of $\mathcal S_H$ to $v$, and inducing an injective homomorphism $p_\#:\pi_1(\mathcal S_H,\tilde v)\to \pi_1(\mathcal S,v)$, such that $p_\#\pi_1(\mathcal S_H , \tilde v) = H$. A covering $p:\mathcal S_H\to \mathcal S$ is called \textit{Galois} or \textit{normal} when $H$ is a normal subgroup of $\pi_1(\mathcal S,v)$. All the coverings that we will study in Section \ref{section:New constructions of quantum Tanner codes} are of this form. It can be shown that, for such a covering map, we have the isomorphism $\operatorname{deck}(p)\cong\pi_1(\mathcal S,v)/H$. The most important result on coverings is the classification theorem known as the Galois correspondence, which we state in our restricted setting of finite index coverings of complexes with finitely many cells.\footnote{This theorem applies in a much broader context, but this is sufficient for our applications.}

\begin{theorem}[Galois correspondence]\label{Theorem Galois correspond}
For all $r\in \mathbb N$, there is a bijection between the set of basepoint-preserving isomorphism classes of index-$r$ connected covering spaces of a complex $\mathcal S$ and the set of index-$r$ subgroups of $\pi_1 (\mathcal S)$. Given such a covering $p:\tilde{\mathcal S}\to \mathcal S$, it is obtained by associating the subgroup $H=p_\#\pi_1 (\tilde{\mathcal S} )$ to the covering space $\tilde{\mathcal S}$. The index of the associated covering is equal to the index $[\pi_1(\mathcal S):H]$.
\end{theorem}
This theorem motivates our exhaustive searches of all possible Galois coverings in Section \ref{section:New constructions of quantum Tanner codes}.

\subsection{Tanner codes and their lifts}\label{section Tanner codes and their lifts}

In this section, we recall some definitions and set up notations related to linear codes. The central components of our quantum codes are classical Tanner codes. We review their construction and a method to lift them into larger codes.\par

\paragraph{Linear codes.} We denote the parameters of a binary linear code $C\subseteq \mathbb{F}_2^n  $  as $[n,k,d]$, where $n$ its the length, $k$ its dimension and $d$ its distance, defined as $d=\min\limits_{c\in C\setminus \{0\} }|c|$. If $k=0$, our convention is to set $d=\infty$. In what follows, a linear code will be defined by the image of a linear map $g:\mathbb F_2^k\to \mathbb F_2^n$,  or the kernel of a linear map, $h:\mathbb F_2^n\to \mathbb F_2^m$, respectively represented in a preferred basis $B$ by a generator matrix $G:=\operatorname{Mat}_B (g)$, and a parity-check matrix $H:=\operatorname{Mat}_B (h)$. Given $x\in \mathbb F_2^n$, the vector $s=Hx$ is called the $syndrome$ of $x$.  The generator matrix of $C$ is the parity-check matrix of its dual $C^\perp=\{x\in \mathbb F_2^n\: :\: \forall c\in C, \langle x,c\rangle=0\}$. A code is called self-dual when $C^\perp=C.$\par

\paragraph{Classical Tanner codes.} The central components of our quantum codes are classical Tanner codes, introduced in \cite{Tanner1981ARA}, and made popular by Sipser and Spielman \cite{Sipser1996}. We first review their construction, and then we describe a natural method to lift them into larger Tanner codes.\par

Let $\mathcal G=(V,E)$ be a graph on $|E|=n$ edges. In the following, $\mathbb{F}_2E :=  \bigoplus_{e\in E}\mathbb{F}_2 e$ denotes the space of formal linear combinations of elements of the set $E$, playing the role of basis vectors. Given a vertex, $v$, the space $\mathbb F_2E(v)$, is the restriction of $\mathbb F_2E $ to the edge neighborhood of $v$. We define a set of binary \textit{local codes} $C_V :=  (C_v )_{ v\in V}$ on the vertices of $\mathcal G$, where an element indexed by $v\in V$ is a code defined on the edge neighborhood of $v$, i.e. $C_v \subseteq \mathbb F_2E(v)$. For a vector $c\in \mathbb F_2 E$, we denote by $c|_{E(v)} $ its restriction to the edge-neighborhood of $v$.

\begin{definition}[Classical Tanner code]\label{definition classical Tanner code}
Let $\mathcal G=(V,E)$ be a graph and $C_V=(C_v )_{ v\in V}$ be a set of binary local codes. We define the Tanner code associated to $\mathcal{G}$ and $C_V$ as
\[ T(\mathcal G,C_V) :=  \{ c\in \mathbb F_2E : c|_{E(v)}\in C_v\text{ for all }v\in V\}.\]
\end{definition}
\paragraph{Parity-check matrix.} A Tanner code can also be defined as the kernel of a parity-check matrix. Suppose that $C_v$ is the kernel of a linear map $h_v:\mathbb F_2E(v)\to \mathbb F_2^{r_v}$ and let $i_v: \mathbb F_2^{r_v}\to \bigoplus_{u\in V} \mathbb F_2^{r_u}$ be the canonical linear embedding map. The Tanner code $C=T(\mathcal G,C_V)$ is the kernel of the map $
h :\mathbb F_2 E\xrightarrow[]{} \bigoplus_{u\in V} \mathbb F_2^{r_u}$ defined on an edge $e$ between vertices $v$ and $w$ by
\begin{equation}\label{equation:tanner code map}
h (e)= i_v\circ h_v (e) + i_w\circ h_w (e),
\end{equation}
and extended by linearity over $\mathbb F_2 E$. All vector spaces being given with a preferred basis, $h_v$ can be represented by a matrix $H_v$ of size $r_v\times |E(v)|$, and $h$ by a matrix $H$ of size $r\times n$, where $r=\sum_{v\in V} r_v$. Denoting $w$, $w_v$, the maximum row-weight of $H$ and $H_v$, respectively, and $q,q_v$ their column-weight, we have \[ w=\max\limits_{v\in V }w_v, \quad \quad q\leq\max \limits_{ \{v,w\}\in E} q_v+q_w .\]
This justifies why this technique is particularly suitable for constructing LDPC codes.

\subsection{Lift of Tanner codes}

Linear code lifting is an operation of major interest to produce interesting LDPC codes \cite{Thorpe2003} in a wide range of block lengths. Recall that lifting amounts to constructing a covering of the Tanner graph representation of the code. In this work, we focus on lifts of Tanner codes. In that case, out of the many possible lifts, there exists a preferred choice that we simply refer to as the \textit{lift of a Tanner code}, as we will not use other types of lifts. From now on, we consider a graph $\mathcal G$ that has no self-loops. Given a graph covering $p:\tilde{\mathcal G}\to \mathcal  G$, where $\tilde{\mathcal G}=(\tilde V,\tilde E)$, the linear extension $p_\#$ of $p$ on the set of edges induces an isomorphism\footnote{In the context of coverings of cell complexes, $p_\#$ acts on chains of  $\mathbb F_2\tilde E$ and $\mathbb F_2\tilde V$ in the cellular chain complex and is a chain map.}
\[ \mathbb F_2\tilde E(\tilde v)\overset{p_\#}{\cong} \mathbb F_2 E(v),\]for any $\tilde v\in p^{-1}(v)$. By definition, $p_\#$ sends basis vectors to basis vectors, i.e. edges to edges. 
\begin{definition}[Lift of Tanner code]\label{definition Lift of classical Tanner code}
 Let $C=T(\mathcal G,C_V)$ and $p:\tilde{\mathcal G}\to \mathcal  G$ be a graph covering, with $\tilde{\mathcal G}=(\tilde V, \tilde E)$. For each vertex $\tilde {v}\in \tilde V$, we associate the code $C_{\tilde v}\subseteq \mathbb F_2E(\tilde v)$, such that $C_{\tilde v} \overset{p_\#}{\cong} C_{p(\tilde v)}$. The lifted Tanner code associated to $p$ is the code $\tilde C=T(\tilde{\mathcal G},C_{\tilde V})$, where $C_{  \tilde V}= (C_{\tilde v})_{ \tilde v\in \tilde V}$. In other words,
\[ T(\tilde{\mathcal G},C_{\tilde V})=\{c\in \mathbb F_2 \tilde E \: :\: p_\#(c|_{\tilde E(\tilde v)})\in C_{p(\tilde v)}, \text{for all }\tilde v\in \tilde V\} .\]
\end{definition}

\paragraph{Parity-check matrix.} Suppose that for each $v\in V$, $C_v$ is the kernel of a map $h_v$ represented by a matrix $H_v$, as described after Equation \eqref{equation:tanner code map}. This induces a parity-check matrix $H$ of size $r\times n$ for the Tanner code $C$. In this context, by lifting we always mean that for any $\tilde v\in p^{-1}(v)$, we define the code $C_{\tilde v}$ as the kernel of a map $h_{\tilde v}$ represented by the matrix $H_{\tilde v}=H_v$. In this way, fixing a matrix representation $H$ induces a matrix representation $\tilde H$, the parity-check matrix of $\tilde C$, up to permutation of the rows. For a covering of index $t$, this matrix has size $tr\times tn$. Denoting $\tilde w$, $\tilde q$ as the maximum row and column-weight of $\tilde H$, we have
\begin{equation}\label{Equation: lift conserve degree}
    \tilde q=q, \quad \tilde w=w.
\end{equation}
This can also be deduced from the fact that the lift of a Tanner code is an example of lift of a classical code associated to a specific type of covering over the Tanner graph associated to $H$.\par

\subsection{Local codes}\label{section:local codes}

Throughout this work, for an integer $\ell\geq 1$, we denote $[\ell] :=   \{ 0,\dots ,\ell-1\}$. 

\paragraph{Classical tensor-product codes.} The classical local codes that will constitute our Tanner codes are built as duals of tensor-product codes. The tensor-product of two linear codes, $C\subseteq \mathbb F_2^{\ell_C}$ and $D\subseteq \mathbb F_2^{\ell_D}$, is the linear code denoted 
\[ \Pi=C\otimes D\subseteq \mathbb F_2^{\ell_C}\otimes \mathbb F_2^{ \ell_D}.\]
The canonical bases of $ \mathbb F_2^{\ell_C}$ and $\mathbb F_2^{\ell_D}$ are indexed by the elements of $[\ell_C]$ and $[\ell_D]$, respectively. These bases naturally induce the canonical product basis for $\mathbb F_2^{\ell_C}\otimes \mathbb F_2^{ \ell_D}$, whose elements are indexed by $[\ell_C]\times[\ell_D]$. In this basis, the codewords of the product code can be seen as the $\ell_C\times\ell_D$ matrices all of whose rows and columns are, respectively, codewords of $C$ and $D$. We may describe a parity-check matrix of $\Pi$ as follows. Let $H_C$ and $H_D$ be parity-check matrices for $C$ and $D$, respectively. Then the codewords of $\Pi$ are the vectors $x\in \mathbb F_2^{\ell_C}\otimes \mathbb F_2^{ \ell_D}$ satisfying both
\[(H_C\otimes I) x=0\quad \text{and}\quad (I\otimes H_D)x=0.\]
The first condition ensures that every row vector of $x$, seen as a $\ell_C\times\ell_D$ matrix, is in $C$, and the second that every column vector is in $D$. A parity-check matrix $H$ of $\Pi$ is hence given by vertically joining these two sets of equations:
\[H=\begin{bmatrix}
    H_C\otimes I\\
    I\otimes H_D
\end{bmatrix}.\]
It follows that the dual of the tensor-product is
\begin{equation}\label{equation dual of tensor product code}
\Pi^\perp=C^\perp\otimes \mathbb F_2^{\ell_D}+  \mathbb F_2^{\ell_C}\otimes D^\perp.
\end{equation}
Indeed, the rows of a parity-check matrix generate the dual code, and the first and second sets of equations of $H$ generate the first and second summands of Equation \eqref{equation dual of tensor product code}, respectively. \\

In the tensor codes considered later, one of the factors is a repetition code of length $2$ and the other one is either a cyclic or a double-circulant code. We therefore recall some essential facts about these two types of codes, taken from \cite{MacWilliamsSloane}.

\paragraph{Cyclic codes.} A binary cyclic code ${C} \subseteq \mathbb F_2^\ell$ is a code that is stable under the action of the automorphism $\rho :\mathbb  F_2^\ell \to  \mathbb F_2^\ell, (c_0, \ldots, c_{\ell-1})\mapsto(c_{\ell-1},c_0,\ldots, c_{\ell-2}) $.\par
A practical way to define a cyclic code is to identify $\mathbb F_2^\ell$ and the polynomial ring $ R_\ell :=  \mathbb F_2[X]/{(X^\ell-1)}$, using the $\mathbb F_2$-linear isomorphism \[ \varphi: R_\ell \longrightarrow \mathbb F_2^\ell,\quad  g_0 + g_1 X + \cdots + g_{\ell-1}X^{\ell-1}  \longmapsto 
  (g_0, \ldots,g_{\ell-1}).\]
In $ R_\ell$, the transformation corresponding to the rotation $\rho$ is multiplication by $X$. A code $C\subseteq R_\ell$ is cyclic if it is stable under this operation, and by linearity it is an ideal of $R_\ell$. These ideals are in one-to-one correspondence with the divisors of $X^\ell-1$. Given $g \in \mathbb F_2 [X]$ such that $g ~|~ X^\ell-1$, the code $\mathscr{C}_\ell (g)$ is defined as the code equal to the ideal generated by $g$, and $g$ is referred to as the \textit{generating polynomial} of the code. It is well-known that $\mathscr{C}_\ell (g)$ has dimension $\ell-\deg (g)$.\par

The dual of a cyclic code is cyclic, and its generating polynomial can be obtained as follows. Given a polynomial $h\in R_\ell$, let $\bar{h}  :=   X^{\deg h}h(1/X)$ denote the reciprocal polynomial of $h$. Over $\mathbb F_2$, $X^\ell - 1$ is equal to its reciprocal polynomial so, if $h~|~X^\ell-1$, then $\bar{h}~|~X^\ell-1$. Letting $h$ be the polynomial such that $gh = X^\ell-1$, also called the check-polynomial, we have  $\mathscr{C}_\ell (g)^{\bot} = \mathscr{C}_\ell (\bar h)$.\par

Given any polynomial $f\in R_\ell$, we define $\mathbb G(f)$ as the circulant matrix whose first row is $\varphi(f)$. Note that $G=\mathbb G(g)$ is a square generator matrix of the code. The row-weight of $G$ is hence directly given by the weight of $g$, i.e. its number of non-zero coefficients.\par

\paragraph{Double-circulant codes.} The second type of linear codes considered in the tensor-product are double-circulant codes. They form a subclass of quasi-cyclic codes, built from any polynomial $f\in R_\ell$. The double-circulant code $\mathscr{D}_{2\ell}(f)\subseteq \mathbb F_2^{2\ell}$ is defined as the image of the generator matrix \[G= \left[\begin{array}{c|c} 
 \mathbb G(1) & \mathbb G(f) \end{array}\right],\]
where $\mathbb G(1)$ corresponds to the $\ell\times\ell$ identity matrix. In Section~\ref{section:Quantum Tanner code with double-circulant local code}, the ambient space $\mathbb F_2^{2\ell}$ of a double-circulant code $\mathscr{D}_{2\ell}(f)$ is always understood to be given with the canonical basis indexed by the columns of the generating matrix in the above form. It is important to note that the code generated by $G$ is invariant under simultaneous cyclic permutation within the left and right blocks. That is, the code is stable under the automorphism 
\begin{equation*}
    \sigma:\begin{cases}
       \hfill \mathbb F_2^{2\ell}\hfill &\longrightarrow  \qquad \mathbb F_2^{2\ell} \\
        (c_0, c_1,\dots,c_{2\ell-1})&\mapsto
(c_{\ell-1},c_0, \ldots, c_{\ell-2},c_{2\ell-1},c_\ell, \ldots, c_{2\ell-2}).
    \end{cases}
\end{equation*}
Hence, $\mathbb Z_\ell$ is a subgroup of the automorphism group $\operatorname{Aut}(\mathscr{D}_{2\ell}(f))$. The dual of $\mathscr{D}_{2\ell}(f)$ is generated by the matrix \[ H=\left[\begin{array}{c|c} 
 \mathbb G(\bar f)&\mathbb G(1)   \end{array}\right],\]
and is also invariant under simultaneous cyclic permutations within the left and right blocks.

\section{Quantum Tanner codes and their lifts}\label{section: quantum Tanner codes and their lifts}

\subsection{Quantum Tanner codes}\label{section Quantum Tanner codes}

\paragraph{Quantum CSS codes.} CSS codes are instances of stabilizer quantum error correcting codes, which first appeared in the seminal work of Calderbank, Shor and Steane \cite{Calderbank1996,Steane1996,Stean1996Multiple}. Let $C_X$ and $C_Z$ be two linear codes with parity-check matrices $H_X$ and $H_Z$, such that $C_X^\perp\subseteq C_Z$, referred to as the \textit{orthogonality condition}. The CSS code $\text{CSS}(C_X,C_Z)$ is the subspace \[\operatorname{Span}\left \{\sum_{z\in C_Z^\perp }\ket{x+z}\text{ } \: :\:\text{ } x\in C_X \right \}\]
of $(\mathbb{C}^2)^{\otimes n}$. Its parameters, denoted $[[n,k,d]]$, are \textit{its length} $n$, \textit{dimension} $k = \dim(C_X / C_Z^\perp)$ and \textit{distance} $d$. The latter is defined as $d := \min(d_X, d_Z)$, with 
\[d_X=\min\limits_{c\in C_Z\setminus C_X^\perp}|c|, \quad 
d_Z=\min\limits_{c\in C_X\setminus C_Z^\perp}|c|.\] The maximum row weight and column weight in the parity check matrix $H_X$, respectively $H_Z$, are denoted $w_X,q_X$, respectively $w_Z, q_Z$. We denote $W:=\operatorname{max}(w_X,q_X,w_Z,q_Z)$. A family of codes is called \textit{Low Density Parity Check} (LDPC) if $W$ is upper bounded by a constant. If $k=0$, our convention is to set $d=\infty$.\par
The only CSS codes that we consider in this article are \textit{quantum Tanner codes}, a class of codes introduced by Leverrier and Z\'emor \cite{Leverrier2022}, in which $C_X,C_Z$ is a pair of classical Tanner codes.\par 
\begin{figure}[t]
    \centering
    \includegraphics[width=0.3\linewidth]{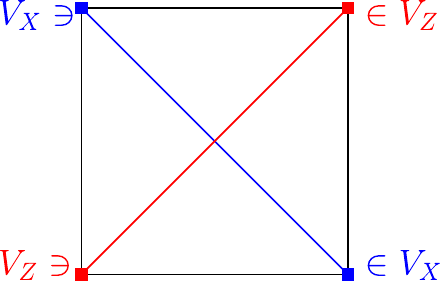}
    \caption{Bipartite square complex $\mathcal S=(V,E,F)$ in which one set of vertices is denoted $V_X$ and the other set $V_Z$. They give rise to two diagonal graphs, $\mathcal G_X^\Box=(V_X,F)$ represented in blue, and $\mathcal G_Z^\Box=(V_Z,F)$ in red.}
    \label{fig: diagonal graphs}
\end{figure}
\paragraph{Square complexes and diagonal graphs.} A practical way to build two Tanner codes $C_X$ and $C_Z$ satisfying the orthogonality condition is to generate them simultaneously by considering two specific graphs embedded in a square complex. Let $\mathcal S :=  (V,E,F)$ be a bipartite square complex and denote the partition $V=V_X \sqcup V_Z$. As remarked in \cite{Leverrier2022}, if we restrict the vertex set to $V_X$, each square face is now incident to only two vertices located at opposite corners. The set of squares $F$ can now be seen as a set of edges on $V_X$, and therefore it defines a graph that we denote by $\mathcal G_X^\Box=(V_X,F)$ and call the \textit{diagonal graph} on $V_X$. Similarly, restricting the vertex set to those of $V_Z$, we obtain the diagonal graph $\mathcal G_Z^\Box=(V_Z,F)$. These graphs, illustrated in Figure~\ref{fig: diagonal graphs}, can then be treated as the graph components of two classical Tanner codes whose set of bits are identified with $F$, and as so we can identify the ambient space of these codes as $\mathbb F_2F$.\par
\begin{figure}[t]
    \centering
    \includegraphics[width=0.4\linewidth]{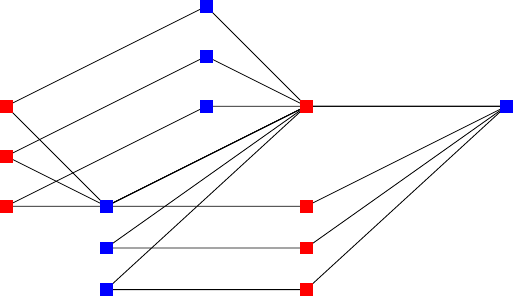}
    \caption{Subgraph of a bipartite square complex with a local product structure. }
    \label{fig:square complex local product}
\end{figure}
\paragraph{Local product structure.} The idea behind quantum Tanner codes is to use a square complex endowed with a specific structure.

\begin{definition}
Let $\mathcal S=(V,E,F)$ be a square complex. We say that $\mathcal S$ has a \textit{local product structure} if, for every vertex $v\in V$, there exist sets $A_v$ and $B_v$, and bijections $(\phi_v)_{v\in V}$ such that $\phi_v:F(v)\to A_v\times B_v$, where $F(v)$ is the face neighborhood of $v$. Moreover, for any pair of adjacent vertices $v$ and $v'$, either $A_v=A_{v'}$ or $B_v=B_{v'}$. In the former case, there are elements $b\in B_v$ and $b'\in B_{v'}$ such that $\phi^{-1}_v(A_v\times \{b\})=\phi^{-1}_{v'}(A_{v'}\times \{b'\})$. In the latter case there are elements $a\in A_v$ and $a'\in A_{v'}$ such that $\phi^{-1}_v(\{a\}\times B_v)=\phi^{-1}_{v'}(\{a'\}\times B_{v'})$.\par
We say that a local product structure is \textit{strict} if $\phi^{-1}_v(\{a\}\times \{b\})=\phi^{-1}_{v'}(\{a\}\times \{b'\})$ for every $a\in A_v$ in the former case, and if $\phi^{-1}_v(\{a\}\times \{b\})=\phi^{-1}_{v'}(\{a'\}\times \{b\})$ for every $b\in B_v$ in the latter case.
\end{definition}

The local product structure of a square complex is illustrated in Figure~\ref{fig:square complex local product}. The left-right Cayley complex of \cite{Dinur2021} can be given a strict local product structure. This was exploited in the quantum Tanner code construction of \cite{Leverrier2022}.

\paragraph{Choice of local codes.} To fully define $C_X$ and $C_Z$, we must coherently assign local codes $C_{V_X}$ and $C_{V_Z}$ to the vertices of the graphs $\mathcal G_X^\Box$ and $\mathcal G_Z^\Box$, respectively. When $\mathcal S$ has a local product structure, it is natural to choose the local codes to be classical tensor product codes. Here, for simplicity we detail the case where the local product structure is strict. However, we will later go beyond this constraint, which will require our local codes to have certain symmetries.\par

Given two adjacent vertices $v_x$ and $v_z$ in such a square complex, let us assume that $\phi_{v_x}:F(v_x)\to [\ell]\times[\ell']$ and $\phi_{v_z}:F(v_z)\to  [\ell]\times [\ell'']$, where $\ell, \ell'$ and $\ell''$ are all positive integers. We can see these bijections as indexings of face neighborhoods. Suppose also that the intersection of their neighborhoods is of the form $[\ell]\times \{i\}$ according to the indexing of $F(v_x)$, and $[\ell]\times \{j\}$ for that of $F(v_z)$. We can extend $\phi_{v_x}$ by linearity into a map $\phi_{v_x}: \mathbb F_2F(v_x)\to \mathbb F_2^{\ell}\otimes \mathbb F_2^{ \ell'}$ mapping faces to canonical basis vectors. We consider local codewords of the product code $C\otimes C' \subseteq \mathbb F_2^{\ell}\otimes \mathbb F_2^{ \ell'}$, where $C$ is a code of length $\ell$ and $C'$ is a code of length $\ell'$. This code can be represented as the set of binary $\ell\times\ell'$ matrices, all of whose rows and columns are elements of $C$ and $C'$, respectively. By extending $\phi_{v_z}$ by linearity into a map $\phi_{v_z}: \mathbb F_2F(v_z)\to \mathbb F_2^{\ell}\otimes \mathbb F_2^{ \ell''}$, we can also consider local codewords of the product code $C^\perp\otimes C''^\perp$ defined in $F(v_z)$, where $C$ is as previously, and $C''$ is a code of length $\ell''$. Codewords can be similarly represented in a $\ell\times \ell''$ matrix. When the local indexing of the faces is strict, the intersection $F(v_x)\cap F(v_z)$ is mapped by $\phi_{v_x}$ and $\phi_{v_z}$ to a column in each of the two matrices. This column supports codewords of $C$ in the first, and codewords of $C^\perp$ in the second, as shown in Figure~\ref{Figure:local product check}. Since we chose a code and its dual, they overlap on an even number of positions in $\tilde F:=F(v_x)\cap F(v_z)$. That is, given $c_x\in C\otimes C'$ and $c_z\in C^\perp\otimes C''^\perp$, we denote $\tilde c_x$ and $\tilde c_z$ the restrictions of $\phi^{-1}_{v_x}(c_x)$ and $\phi^{-1}_{v_z}(c_z)$ to $\mathbb F_2 \tilde F$, respectively. By our choice of local codes, we automatically have $\langle \tilde c_x,\tilde c_z\rangle=0$. We can hence define the local codewords $\phi^{-1}_{v_x}(c_x)$ and $\phi^{-1}_{v_z}(c_z)$ as parity-check constraints associated with the code $C_{v_x}$ and $C_{v_z}$, respectively.\par
\begin{figure}[t]
\begin{center}
\begin{tikzpicture}[scale=0.6]

% Light grey column (background for first grid)
\fill[gray!20] (2,0) rectangle (3,8);

% First grid (left)
\foreach \x in {0,...,8} {
  \draw[black!30] (\x,0) -- (\x,8);
}
\foreach \y in {0,...,8} {
  \draw[black!30] (0,\y) -- (8,\y);
}

% Labels for first grid
\node at (4,8.5) {\textbf{X-checks of $v_x$}};
\node[rotate=0] at (-0.7,4) {\textbf{$C$}};
\node at (4,-0.7) {\textbf{$C'$}};

% Vector values for first grid
\foreach \i/\val in {0/1, 1/0, 2/0, 3/0, 4/0, 5/1, 6/1, 7/1} {
  \node at (2.5,7.5-\i) {\val};
}

% Second grid (right)
\begin{scope}[shift={(10,0)}]
  % Light grey column (background for second grid)
  \fill[gray!20] (3,0) rectangle (4,8);

  % Grid lines
  \foreach \x in {0,...,8} {
    \draw[black!30] (\x,0) -- (\x,8);
  }
  \foreach \y in {0,...,8} {
    \draw[black!30] (0,\y) -- (8,\y);
  }

  % Labels for second grid
  \node at (4,8.5) {\textbf{Z-checks of $v_z$}};
  \node[rotate=0] at (-0.7,4) {\textbf{$C^\perp$}};
  \node at (4,-0.7) {\textbf{$C''^\perp$}};

  % Vector values for second grid
  \foreach \i/\val in {0/0, 1/1, 2/0, 3/0, 4/1, 5/0, 6/1, 7/1} {
    \node at (3.5,7.5-\i) {\val};
  }
\end{scope}

\end{tikzpicture}
\end{center}
\caption{Face neighborhood $F(v_x)$ and $F(v_z)$, identified with two product sets and represented as matrices. The gray column in the left array and in the right array corresponds to identical faces in the square complex. A $X$-check can be seen as a codeword of the product code $C\otimes C'$ in the left array, and a $Z$-check as a codeword of the product code $C^\perp \otimes C''^\perp$ in the right array. The restriction of the two codewords to the gray columns are orthogonal vectors, ensuring that the $X$-check and the $Z$-check commute.}
\label{Figure:local product check}
\end{figure}

When making such a consistent choice of local product codes across the entire square complex, we obtain two Tanner codes 
\[ C_X=T(\mathcal G_X^\Box, C_{V_X} ), \quad C_Z=T(\mathcal G_Z^\Box, C_{V_Z}),\]
such that $C_X^\perp\subseteq C_Z$. Such a pair of classical codes defines a quantum Tanner code denoted $T(\mathcal S, C_{V_X}, C_{V_Z})$.\par

The procedure above is better carried out case by case, especially when the local product structure is not strict. In that case, the assignment of the local codes depends on the global structure of the square complex. We illustrate it on explicit examples in Section~\ref{section:New constructions of quantum Tanner codes}.\par

\subsection{Lifting quantum Tanner codes}\label{section lift of quantum Tanner codes}

We now describe a lifting procedure for quantum Tanner codes which differs from that of \cite{GuemardLiftIEEE}. Our approach is an extension of the method applied in \cite{Panteleev2021} to design a family of Sipser-Spielman codes \cite{Sipser1996} with optimal parameters. Here, we use a covering of a bipartite square complex in order to create a covering for each of its two diagonal graphs.\par

Suppose that $\operatorname{CSS}(C_X, C_Z)$ is a quantum Tanner code, built over a bipartite square complex $\mathcal S=(V=V_X\sqcup V_Z,E,F)$, with $C_X=T(\mathcal G_X^\Box, C_{V_X} )$ and $C_Z=T(\mathcal G_Z^\Box, C_{V_Z})$. Given a covering map $p:\tilde{\mathcal S} \to \mathcal S$, the 1-skeleton of $\tilde{\mathcal S}=(\tilde V,\tilde E,\tilde F)$ is also bipartite with bipartition $\tilde V=\tilde V_X\sqcup \tilde V_Z$, where $\tilde V_X=p^{-1}(V_X)$ and $\tilde V_Z=p^{-1}(V_Z)$. We can hence consider the diagonal graphs $\tilde{\mathcal G}_X^\Box=(\tilde V_X,\tilde F)$ and $\tilde{\mathcal G}_Z^\Box=(\tilde V_Z,\tilde F)$. \par

Using the notation above, we have the following lemma.
\begin{lemma}\label{lemma restriction covering map}
The covering map $p$ induces two covering maps $p_X: \tilde{\mathcal G}_X^\Box \to \mathcal G_X^\Box $ and $p_Z:\tilde{\mathcal G}_Z^\Box \to \mathcal G_Z^\Box$.
\end{lemma}
\begin{proof}
Suppose that $p:\tilde{\mathcal S}\rightarrow  \mathcal S$ is an index-$t$ covering defined by an assignment of permutations to all the oriented edges of $\mathcal S$, i.e $\pi_e\in S_t$ for $t>0$ such that $\pi_{e^{-1}}=\pi_{e}^{-1}$. It is essential to note that $\tilde V_X=p^{-1}(V_X)$, $\tilde V_Z=p^{-1}(V_Z)$, and $\tilde F =p^{-1}(F)$, so that we have $t$ vertices or faces in $ \tilde{\mathcal G}_X^\Box$ projected by $p$ onto each vertex or face in $ \mathcal G_X^\Box$, and similarly for the diagonal graph on $V_Z$. We now need to show that these cells are assembled into coverings of the diagonal graphs.\par
 Given a face $f\in F$ described by the closed path of oriented edges $e_1.e_2.e_3.e_4=[v_1,v_2].[v_2,v_3].[v_3,v_4].[v_4,v_1]$, then for all $i\in [r]$, the face labeled $(f,i)$ in $\tilde{\mathcal S}$ can be defined by the 4-cycle as described in Section~\ref{section covering maps}. Therefore, if $e=[v_1,v_3]$ is an edge in $\mathcal G_X^\Box$, then an edge $(e,i)$ in the covering $\tilde{\mathcal G}_X^\Box$, projected onto $e$ by $p$, is given by $[(v_1,i), (v_3,\pi_{e_2}\circ \pi_{e_1}(i))]$. Let us assign the permutation $\pi_{e_2}\circ \pi_{e_1}\in S_t$ to $e$. We can similarly assign a composition of two permutations to each diagonal edge of $\mathcal G_X^\Box$, according to the corresponding square in $\mathcal S$. This defines a covering map $p_X: \tilde{\mathcal G}_X^\Box \to \mathcal G_X^\Box $. A similar procedure on $\tilde{\mathcal G}_Z^\Box$ defines a covering map $p_Z:\tilde{\mathcal G}_Z^\Box \to \mathcal G_Z^\Box$.
\end{proof}

We now have all the elements to define the lift of a quantum Tanner code. Notice that we can define two lifted classical Tanner codes from our square complex covering. Let $\tilde C_X=T(\tilde{\mathcal G}_X^\Box ,C_{\tilde V_X})$ be the lift of $C_X=T(\mathcal G_X^\Box, C_{V_X} )$ associated to $p_X$, and $\tilde C_Z=T(\tilde{\mathcal G}_Z^\Box ,C_{\tilde V_Z})$ be the lift of $C_Z=T(\mathcal G_Z^\Box, C_{V_Z})$ associated to $p_Z$. The codes $\tilde C_X$ and $\tilde C_Z$ are subspaces of a common vector space $\mathbb F_2^{t\cdot n}$, where $t$ is the index of the cover, and they are candidates for defining a new CSS code. This is conditioned on the fact that $\tilde C_X$ and $\tilde C_Z$ must satisfy the orthogonality condition, namely $\tilde C_X^\perp\subseteq \tilde C_Z$. If that condition is satisfied, we define the lifted quantum Tanner code associated to $p$ as the code $\operatorname{CSS}(\tilde C_X,\tilde C_Z)$.\par

This definition suggests that the lifted classical Tanner codes associated to $p_X$ and $p_Z$ do not necessarily satisfy the orthogonality condition. This depends on certain properties of the complex $\mathcal S$, and on the choice of local codes. However, there are cases where they always do, as shown below.

\begin{proposition}\label{proposition quantum Tanner lifting simply connected closure}
Let $\operatorname{CSS}(C_X,C_Z)$ be a quantum Tanner code over a bipartite square complex $\mathcal{S}$ for which the face neighborhood of every vertex has a simply connected closure\footnote{The closure of a set of faces is the subcomplex composed of those faces along with the set of vertices and edges incident to them.}, when seen as a cell complex. Then, the lifting procedure above, applied to every finite covering space $p:\tilde{\mathcal{S}}\to \mathcal{S}$, defines a valid CSS code.
\end{proposition}

\begin{proof}
    The face neighborhood of every vertex having a simply connected closure, the support intersection between a check of $\tilde C_X$ and one of $\tilde C_Z$, associated to two incident vertices of $\tilde{\mathcal S}$, is isomorphic to the support intersection of their projections in $\mathcal S$ by the covering map.
\end{proof}
In Figure~\ref{fig: closure face neighborhood}, we show an example of a vertex in a square complex whose face neighborhood has a closure that is not simply connected.
\begin{figure}[t]
    \centering
    \includegraphics[width=0.28\linewidth]{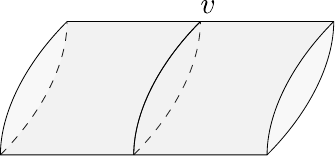}
    \caption{Face neighborhood of a vertex $v$ in the topological realization of a square complex. It is composed of four faces and its closure forms a subcomplex that is not simply connected: a cylinder.}
    \label{fig: closure face neighborhood}
\end{figure}
The condition of Proposition~\ref{proposition quantum Tanner lifting simply connected closure} is satisfied for a wide range of quantum Tanner codes, including those built from square complexes which are products of two graphs, or left-right Cayley complexes satisfying the total-no-conjugacy-class condition \cite{Dinur2021}. It was shown in \cite{GuemardLiftIEEE}, that the quadripartite version of the left-right Cayley-complex of \cite{Dinur2021} can be obtained as a covering of a Cartesian product of graphs. This can be used to show that the quantum Tanner code family defined in \cite{Leverrier2022,leverrier2022decoding} on such left-right Cayley-complex can be obtained by quantum Tanner code lifting, similarly to the approach of \cite[Section IV.A]{GuemardLiftIEEE}.\par

Hereafter, we define quantum Tanner code families beyond the restriction on $\mathcal S$ stated in Proposition~\ref{proposition quantum Tanner lifting simply connected closure}. It means that the number of qubits shared by a pair of $X$- and $Z$-checks in the base codes is not necessarily preserved under lifting. In particular, this shows that the lifting technique presented here is not equivalent to that of \cite{GuemardLiftIEEE}, where lifting a code preserved the size of the overlap between any pairs of $X$- and $Z$-checks. In Section \ref{section:New constructions of quantum Tanner codes}, we give explicit examples where this phenomenon occurs. \\

In the rest of this article, we focus on specific types of square complex coverings.

\begin{lemma}
Let $\mathcal S$ be a connected square complex for which every edge is incident to at least two faces, and let  $p:\tilde{\mathcal S}\to \mathcal S$ be covering.
\begin{enumerate}
    \item If $p$ is connected, then $p_X$ and $p_Z$ are also connected coverings.
    \item If $p$ is a Galois covering with group of deck transformations $\Gamma$, then $p_X$ and $p_Z $ are also Galois with group of deck transformations $\Gamma$.
\end{enumerate}     
\end{lemma}

\begin{proof}
\begin{enumerate}
    \item If every edge in $\mathcal S$ is incident to at least two faces, then its diagonal graphs are also connected. Indeed, if two faces are adjacent in $\mathcal S$, then their induced edges in either of the two diagonal graphs must also be adjacent, and by repeating this argument for any two distant faces in $\mathcal S$ shows that any two edges in the diagonal graphs are connected by a path. This local property is preserved under taking a covering, showing that the diagonal graphs of $\tilde{\mathcal{S}}$ are also connected. 
    \item If $p:\tilde{\mathcal S}\to \mathcal S$ is a Galois covering, the diagonal graphs in $\tilde{\mathcal{S}}$ must be connected as a consequence of the first item. Moreover, $\tilde{\mathcal G}_X^\Box$ inherits the free transitive action of the group $\operatorname{deck}(p)$, since we can see $\tilde{\mathcal G}_X^\Box$ as a subcomplex of $\tilde{\mathcal S}$ and $p_X$ as a restriction of $p$ to this subcomplex. Therefore, $p_X: \tilde{\mathcal G}_X^\Box\to \mathcal G_X^\Box$ is also a Galois covering with a group of deck transformations $\operatorname{deck}(p)$, and similarly for $p_Z$.
\end{enumerate}
\end{proof}
In this work, we say that the lift of a CSS code is \textit{Galois} if it is associated to a Galois covering.\par

An automorphism of a quantum code $C=\operatorname{CSS}(C_X,C_Z)$ is defined as a permutation of the coordinates of the qubits which is equivalent to a permutation of the parity checks of both $H_X$ and $H_Z$. The set of automorphisms of $C$ forms the group $\operatorname{Aut}(C)$. In the case of a lifted quantum Tanner code $\tilde{C} $, notice that the group of deck transformations of $p:\tilde{\mathcal S}\to \mathcal  S$ is a symmetry group of the Tanner graph of the quantum code. It is also a symmetry group of each of the Tanner graphs associated to $\tilde H_X$ and $\tilde H_Z$. Thus, given an element\footnote{The action of $\gamma$ is naturally extended to the $\mathbb F_2$-space of the lifted codes.} $\gamma \in \operatorname{deck}(p)$, a vector $c\in \mathbb F_2 \tilde F$ with syndrome $\tilde H_Xc=s_X$ and $\tilde H_Zc=s_Z$ , there exist two permutations $\sigma_X$ and $\sigma_Z$ of the syndrome coordinates such that $\tilde H_X\gamma(c)=\sigma_X(s_X)$ and $\tilde H_Z\gamma(c)=\sigma_Z(s_Z)$. Therefore, $\operatorname{deck}(p)$ is a subgroup of $\operatorname{Aut}(\tilde{C})$. We will use this feature in Section~\ref{section transfer homomorphism} to characterize the logical operators of certain lifted codes.

\subsection{Parameter bounds from transfer homomorphism}
\label{section transfer homomorphism}

For a lift defined from a Galois cover of odd index, we can use our information on the base code and leverage covering theory to lower bound the dimension of the lifted code and to characterize the structure of the logical operators. To this end, we use the language of chain complexes, and a specific chain map called the \textit{transfer homomorphism} \cite{HatcherTopo} usually arising in the context of cellular homology, or more generally in group theory. Although the discussion and parameters bounds presented here holds for code lifting in general, such as the method of \cite{GuemardLiftIEEE}, we focus on lifted classical and quantum Tanner codes. Our findings can be straightforwardly applied to the codes of Example~\ref{example BS2}, and show that logical operators are invariant under cyclic permutations. More formally, the aim of this section is to prove the following bounds on the parameters of a lifted code.

\begin{proposition}\label{proposition parameter bound}
    Let $\tilde C$ be a quantum Tanner code obtained by a Galois $t$-lift, with $t$ odd, of a quantum Tanner code $C$ with parameters $[[n,k,d]]$. Then, the parameters $[[\tilde n,\tilde k,\tilde d]]$ of $\tilde C$ satisfy $\tilde n=tn$, $\tilde k\geq k$, and $\tilde d\leq td$. Moreover, if $\tilde k=k$, then $\tilde d\geq d$.
\end{proposition}

\begin{remark}
    The validity of Proposition \ref{proposition parameter bound} extends to the code lifting method of \cite{GuemardLiftIEEE}, as made clear from the technics, which does not depend on the structural specificities of quantum Tanner codes. In particular, the result just relies on the lifted quantum code $\tilde C$ to be a chain complex of free modules over a group $\Gamma$ of order $t$.
\end{remark}

\paragraph{Classical and quantum codes as chain complexes.} Recall that a classical code can be seen as the first homology group of the chain complex $C_1\xrightarrow[]{h} C_0:=\mathbb F_2^n\xrightarrow[]{h} \mathbb F_2^m$. By abuse of language, we directly refer to the chain complex as a classical code. Moreover, given a CSS code $\mathrm{CSS}(C_X,C_Z)$ the orthogonality between $C_X$ and $C_Z$ is equivalent to $H_X H_Z^T=0$. This \textit{orthogonality condition} is analogous to the composition property of two boundary maps in a chain complex. Since a CSS code is fully specified by two parity-check matrices, it can naturally be represented by a \textit{3-term chain complex} (or its dual),
\begin{equation}
C\coloneq\mathbb{F}_2^{m_Z}\xrightarrow{\partial_{2}=h_Z^*}\mathbb{F}_2^{n}\xrightarrow{\partial_{1}=h_X} \mathbb{F}_2^{m_X}.
\end{equation}
Here, $n$ in the number of physical qubits, while $m_X$ and $m_Z$ are the number of rows in $H_X$ and $H_Z$, respectively. Conversely, any CSS code defines a chain-complex of finite dimensional $\mathbb F_2$-vector-spaces, and we obtain a correspondence between these 2 sets of objects. It follows that the homology group $H_1(C)$ is equal to $C_X/C_Z^\perp$ and corresponds to logical $Z$-operators. Similarly, the co-homology group $H^1(C)$ is equal to $C_Z/C_X^\perp$ and corresponds to logical $X$-operators.\par

\paragraph{Classical Tanner codes.} For linear codes, lifting may induce the presence of symmetries in $T(\tilde{\mathcal G},C_{\tilde V})$. An automorphism of a binary linear code is a permutation of the coordinates mapping the code to itself. The set of automorphisms of a code $C$ forms a group denoted $\operatorname{Aut}(C)$. This feature can be conveniently formulated in the language of chain complexes. Lifting a linear code with the method of \cite{GuemardLiftIEEE} yields a length 2-chain complex which is a free module over the group of deck transformations For lifts of classical Tanner codes, we similarly obtain the following.\par

\begin{lemma}\label{lemma classical is free module}
Let $C :=  C_1\xrightarrow[]{h}C_0$ be a classical (Tanner) code and $\tilde C$ a connected Galois lift of $C$ with group of deck transformations $\Gamma$. Then, $\tilde C$ is a chain complex of free $\Gamma$-modules.
\end{lemma}

\begin{proof}
Given $C$ a Tanner code defined from a graph $\mathcal G$, let $p:\tilde{\mathcal G}\to \mathcal G$ be the Galois covering associated to lifted code $\tilde C$. Now, by fixing an arbitrary parity-check matrix $H:=\text{Mat}(h)$, and given a row $r_i$ of $H$, i.e a check localized at a vertex $v$, we may associate $|\Gamma|$ lifted checks of the lifted parity-check $\tilde H$, one denoted $\tilde r_i$ for each vertex $\tilde v\in p^{-1}(v)$. We may refer to this set of elements as the fiber above $r_i$ and notice that $\Gamma$ acts fiberwise, freely and transitively. Moreover, given an edge $e=[v,v']$, and a lift $\tilde e=[\tilde v,\tilde v']$, we have that $\tilde r_i$ acts on the bit associated to $\tilde e$ if and only if $r_i$ acts on the bit associated to $r_i$ From this, it is clear that $\Gamma$ is a symmetry group of the Tanner graph associated to the lifted parity-check matrix $\tilde H$, and $\Gamma$ acts freely and transitively both on the fiber above the bit and check vertices. Therefore, by extending the action of $\Gamma$ by linearity over elements of $C_1$ and $C_0$, these spaces become free $\Gamma$-modules. In addition, given an element $\gamma \in \operatorname{deck}(p)$ and a vector $x\in \mathbb F_2 \tilde E$ with syndrome $\tilde Hx=s$, we have $\tilde H\gamma(x)=\gamma(s)$. In the language of chain complexes, it means that the action of $\Gamma$ on the constituent spaces commutes with boundary maps.
\end{proof}

A lift induces a chain map, i.e. a map commuting with boundary operators, which we denote $\pi_\#: \tilde C_i \to C_i$. This map sends a basis vector $\tilde x\in \tilde C_i$ of the lifted code to its projection in the base code. To this map corresponds a cochain map between the dual chain complexes, $\pi^\#:C^i\to \tilde C^i$.\par

Notice also that there exists a chain map from our base code to the lifted code, denoted $\tau_\#:  C_i \to \tilde C_i$, and sending a basis vector in $C_i$ to the sum of its $t$ distinct lifts in $\tilde C_i$, namely, for $x\in C_i$, 
\[\tau_\#(x)=\sum_{\tilde x\in p^{-1}(x)} \tilde x.\]
The composition $\pi_\#\circ \tau_\#$ is multiplication by $t$, and since we consider $\mathbb F_2$-vector spaces, this map is the zero map when $t$ is even.\par

\paragraph{Quantum Tanner codes.} We now show how to extend these maps to chain maps on a lifted quantum code. While here we focus on quantum Tanner code, it is valid also for the method of \cite{GuemardLiftIEEE}. This approach provides valuable information only when $t$ is odd, and can be applied to codes of Table~\ref{Table:lift BS(l)}. \par
We consider a quantum Tanner code $C=T(\mathcal S, C_{V_X}, C_{V_Z})$, interpreted as a chain complex $C :=  C_{2}\xrightarrow[]{h_Z^*} C_{1}\xrightarrow[]{h_X} C_{0}$ defined by two classical Tanner codes $C_X=T(\mathcal G_X^\Box, C_{V_X} )$ and $C_Z^*=T(\mathcal G_Z^\Box, C_{V_Z})$, seen as the chain complexes $C_X :=   C_{1}\xrightarrow[]{h_X} C_{0}$ and $C_Z :=  C_{2}\xrightarrow[]{h_Z^*} C_{1}$, respectively, with $\mathrm{Mat}(h_X)=H_X$ and $\mathrm{Mat}(h_Z^*)=H_Z^T$. Here, we interpret the cochain complex of the code $C_Z$ as a chain complex. By performing a lift of $C$ corresponding to an index-$t$ covering $p:\tilde{\mathcal S}\to \mathcal S$, we obtain $t$-lifts $\tilde C_X, \tilde C_Z$ for each of the constituent codes, with $\tilde h_X\circ \tilde h^*_Z=0$. From the previous discussion on the structure of the constituent classical codes, we obtain the following.\par

\begin{lemma}
Let $C$ be a quantum (Tanner) code and $\tilde C$ a connected Galois lift of $C$ with group of deck transformations $\Gamma$. Then, $\tilde C$ is a chain complex of free $\Gamma$-modules.
\end{lemma}
\begin{proof}
From Lemma \ref{lemma classical is free module} and Section~\ref{section lift of quantum Tanner codes}, when $p$ is Galois the spaces of $\tilde C$ inherit the free and transitive action (on the fiber) of the group of deck transformations $\Gamma$. Since this action commutes with boundary map of $\tilde C_X$ and $\tilde C_Z$, it also commutes with the boundary maps of $\tilde C$.
\end{proof}

We also obtain two chain maps, 
\[
\begin{tikzcd}[column sep=large]
\tilde C_{2} \arrow[r, "\tilde h_Z^*" ] \arrow[d, "\pi_\#"] & \tilde C_1 \arrow[d, "\pi_\#"] & \quad & 
\tilde  C_{1} \arrow[r, "\tilde h_X"] \arrow[d, "\pi_\#"] & \tilde  C_0 \arrow[d, "\pi_\#"] \\
C_2 \arrow[r, "h_Z^*"] & C_1 & \quad & 
C_1 \arrow[r, "h_X"] & C_0
\end{tikzcd}.
\]
Notice that these maps have the same action on the space $\tilde C_1$. This also holds for the associated cochain map denoted $\pi^\#$. Similarly, from code lifting, we obtain two chain maps in the other direction, $\tau_\#: C_X\to \tilde C_X$ and $\tau_\#: C_Z\to \tilde C_Z$, having the same action on the space $C_1$. This also holds for the associated cochain map $\tau^\#$. Therefore, we obtain a chain map $\pi_\#$ from the lifted CSS code to the base CSS code, and a chain map $\tau_\#$ in the other direction,
\[
\begin{tikzcd}
\tilde C_2 \arrow[r, "\tilde h_Z^*"] \arrow[d, "\pi_\#"] & \tilde C_1 \arrow[r, "\tilde h_X"] \arrow[d, "\pi_\#"] & \tilde C_0 \arrow[d, "\pi_\#"] &\quad & \tilde C_2 \arrow[r, "\tilde h_Z^*"] & \tilde C_1 \arrow[r, "\tilde h_X"]  & \tilde C_0 \\
C_2 \arrow[r, "h_Z^*"] & C_1\arrow[r, "h_X"] & C_0 &\quad & C_2\arrow[u, "\tau_\#"] \arrow[r, "h_Z^*"] & C_1\arrow[u, "\tau_\#"] \arrow[r, "h_X"] & C_0\arrow[u, "\tau_\#"].
\end{tikzcd}
\]
Furthermore, the composition of these two chain maps $\pi_\#\circ \tau_\#$ remains multiplication by~$t$.\par
These maps and cochain maps induce homomorphisms at the level of homology and cohomology, which we denote
\begin{align*}
    &\pi_*:H_i(\tilde C)\to H_i( C), \qquad \tau_*:H_i( C)\to H_i( \tilde C),\\
    &\pi^*:H^i( C)\to H^i( \tilde C), \qquad \tau^*:H^i( \tilde C)\to H^i(  C),
\end{align*}
where all homology groups are understood to be with $\mathbb F_2$-coefficients. In the context of coverings of cell complexes, $\tau_*$ and $\tau^*$ are referred to as \textit{transfer homomorphisms}. For an element $\gamma \in \Gamma$, we also denote by $\gamma^*$ the induced action of $\gamma$ on homology.\par
By studying the compositions $\pi^*\circ \tau^*$ and $\tau^*\circ \pi^*$, we may obtain information on the dimension of $C$ and on the structure of its logical operators (homology classes). The following proposition is adapted from \cite[Proposition 3G.1]{HatcherTopo}, where it is stated for homology of cell complexes.

\begin{proposition}\label{proposition transfer homomorphism}
 Let $p:\tilde{\mathcal S}\to \mathcal S$ be a Galois covering of odd index $t$, with group of deck transformations denoted $\Gamma$. Then the map $\pi^*: H^i(C)\to  H^i(\tilde C)$ is injective. Its image is the subgroup $H^i(\tilde C)^\Gamma$ consisting of classes $\alpha$, such that for $\gamma\in \Gamma$, $\gamma^*(\alpha)=\alpha$.
\end{proposition}

\begin{proof}
We have all the ingredients to study transfer homomorphisms: a covering $p$ inducing a free transitive (on the fiber) linear action of the group $\Gamma\cong \mathrm{deck}(p)$ on the chain complex $C$ associated with our code, and chain maps $\pi_\#$ and $\tau_\#$ with composition equal to multiplication by $t$. We may therefore revisit the proof of \cite[Proposition 3G.1]{HatcherTopo} and verify that no additional conditions are needed.\par

To study the kernel of $\pi^*$, recall that the composition $\pi_\#\circ \tau_\#$, and hence also $\tau^\#\circ\pi^\#$, acts on a chain by multiplication by $t$. Therefore, $\tau^*\circ\pi^*:H^i(C)\to H^i(C)$ acts on a class by sending $\alpha\mapsto t\alpha$. Supposing $\pi^*(\alpha)=0$, we have $\tau^*\circ \pi^*(\alpha)=t\alpha=0$, and since we assume that $t$ is odd, we must have $\alpha=0$. This shows that $\ker(\pi^*)=0$ and that $\pi^*$ is injective.\par

To describe the image of $\pi^*$, we study the composition $\tau_\#\circ \pi_\#$, which sends a basis vector $\tilde x$ in the lifted code to the sum of all basis vectors in its orbit, namely $\tau_\#\circ \pi_\#(x)=\sum_{\gamma\in \Gamma} \gamma(x)$. The induced map on cohomology acts on a class $\alpha\in H^i(\tilde C)$ as $\pi^*\circ \tau^*(\alpha)=\sum_{\gamma\in \Gamma} \gamma(\alpha)$. To prove that $H^i(\tilde C)^\Gamma\subseteq \mathrm{Im}(\pi^*)$, suppose $\alpha  \in H^i(\tilde C)^\Gamma$, namely that $\gamma^*(\alpha)=\alpha$. Then, $\pi^*\circ \tau^*(\alpha)=t\alpha=\alpha$ since $t$ is odd, and hence $\alpha\in \mathrm {Im}(\pi^*)$. To prove that $\mathrm{Im}(\pi^*)\subseteq H^i(\tilde C)^\Gamma$, notice that $\pi_\#\circ \gamma= \pi_\#$ for all $\gamma\in \Gamma$. Therefore, for any class $\alpha \in H^i(C)$, we have that $\gamma^*\circ\pi^*(\alpha)=\pi^*(\alpha)$, and hence the class $\pi^*(\alpha)$ is invariant under the action of $\Gamma$. 
\end{proof}
We can repeat the same work to prove that the map $\tau_*:H_i(C)\to H_i(\tilde C)$ is injective, with image equal to the subgroup of $H_i(\tilde C)$ consisting of homology classes invariant under the action of $\Gamma$.\par
This proposition directly gives the bounds on the dimension and distance of lifted quantum code, as stated in Proposition~\ref{proposition parameter bound}.

\begin{proof}[Proof of Proposition~\ref{proposition parameter bound}]
The lower bound on the dimension of $\tilde C$ follows directly from the injectivity of $\pi_*$ (and $\tau_*$). For the distance, consider a class $\alpha\in H_i(C)$ represented by a vector $x$ of Hamming weight $d$. Then, $\tau_\#(x)$ is a vector of Hamming weight $td$, and this represents a non-trivial homology class $\tilde \alpha=\tau_*(\alpha)$ in $H_i(\tilde C)$. Therefore, we conclude that the distance of the code $\tilde C$ is upper bounded by $td$. A similar upper bound is obtained by considering the image of a cohomology class under $\pi^*$. Finally, when $\tilde k=k$, the maps $\tau_*$ and $\pi^*$ become isomorphism, which means that a non-trivial (co)homology class $\tilde \alpha$ of $\tilde C$, represented by a vector $\tilde x$ of weight $w$, must be mapped onto a representative $\tilde x$ of a non-trivial (co)homology class $\alpha$ of $C$ under $\pi_\#$ or $\tau^\#$. By covering property, this mapping must decrease the weight, and therefore $\tilde d\geq d$.\par  
\end{proof}

\section{New constructions of quantum Tanner codes}\label{section:New constructions of quantum Tanner codes}

\subsection{General procedure}\label{section:general procedure}

In this section and the next one, we introduce new codes constructed by lifting short quantum Tanner codes. As described previously, to define a code, we need a square complex and a set of local codes associated with its vertices. The possibility of lifting a code is governed by the topology of its square complex, as stated in Theorem~\ref{Theorem Galois correspond}. For the purpose of generating families of lifted codes, we seek a square complex having a fundamental group with infinitely many finite-index subgroups. Our strategy is to select a group $G$ and construct a bipartite square complex $\mathcal S$ whose fundamental group is $G$, and has a local product structure. In Section~\ref{section:Quantum Tanner code with cyclic local code} and Section~\ref{section:Quantum Tanner code with double-circulant local code}, we obtain $\mathcal S$ by cellulating the presentation complex of a finitely-presented group, or a homotopy equivalent space. \par

\paragraph{Presentation complex.} For completeness, we recall the construction of a presentation complex, as exposed in \cite{HatcherTopo}. Let $G=\langle S |  R\rangle$ be a group presentation, where $S$ is the set of generators denoted $g_i$ and $R$ the set of relations denoted $r_j$. The group $G$ is the quotient of a free group $F_{|S|}$ on the generators of $S$ by the normal closure of the group generated by $R$. The relations of $R$ are the generators of the kernel of the map $F_{|S|}\to G$. For any group presentation, we can utilize this quotient structure to construct a two-dimensional cell complex $M_G$, called the \textit{presentation complex} of $G$, having one 0-cell, exactly $|S|$ 1-cells and $|R|$ 2-cells, and such that $\pi_1(M_G)=G$. To construct it, we start from its 1-skeleton: a wedge of circles $\vee_i S_i^1$ attached to a vertex $v$, namely directed edges denoted $e_{g_i}$ with both endpoints attached to $v$. This gives a space with fundamental group $F_{|S|}$. A relation is a product of the generators and it specifies a closed path in the graph. For each relation $r_j$, we attach one 2-cell labeled $f_j$ along the closed path specified by $r_j$. For example, if $r_j=g_i g_j\dots g_k$, then we attach a face along the closed path $e_{g_i}\cdot e_{g_j}\dots e_{g_k}$. The effect is to trivialize the element of $F_{|S|}$ corresponding to~$r_j$. \par 

\paragraph{Construction of square complexes.} We will focus on group presentations with two generators and one relation. As such, all of their associated presentation complexes have two edges and one face.  It is, however, important to note that these 2D spaces are not square complexes in general. Nevertheless, given such a group, by subdividing the face of its presentation complex, or a homotopy equivalent space, into squares, we will be able to obtain a bipartite square complex $\mathcal S=(V=V_X\sqcup V_Z,E,F)$ that has locally the structure of a product space.\par

\paragraph{Base quantum Tanner code.} From the diagonal graphs of $\mathcal S$, and a choice of local linear codes, it is possible to generate two Tanner codes $C_X=T(\mathcal G_X^\Box, C_{V_X} )$ and $C_Z=T(\mathcal G_Z^\Box, C_{V_Z})$, satisfying $C_X^\perp\subseteq C_Z$. Their local codes must be chosen to have a symmetry group adapted to the periodic boundary conditions of $\mathcal S$. These will be duals of tensor products of cyclic or double-circulant codes. \par

\paragraph{Lifted codes.} Finally, we construct new square complexes by considering connected Galois covering spaces, such as $p:\tilde{\mathcal S}\to \mathcal S$. By restricting $p$ to $\mathcal G_X^\Box$ and $\mathcal G_Z^\Box$, we define the covering maps $p_X: \tilde{\mathcal G}_X^\Box \to \mathcal G_X^\Box $ and $p_Z:\tilde{\mathcal G}_Z^\Box \to \mathcal G_Z^\Box$, which are also connected and Galois, as discussed in Section~\ref{section lift of quantum Tanner codes}. We then use the procedure of Section~\ref{section lift of quantum Tanner codes} to lift $\operatorname{CSS}(C_X,C_Z)$ into $\operatorname{CSS}(\tilde C_X,\tilde C_Z)$. Recall that for each lift, the group of deck transformations is a subgroup of the automorphism group of the corresponding code.\par

\paragraph{Numerical method.} To compute subgroups and quotient groups of $G$, needed for the Galois covering-space constructions, we use the GAP package LINS \cite{GAP4}.  As we are interested in short codes, we only compute normal finite-index subgroups up to index 30. To compute an upper bound on the distances $d_X $, $d_Z$ of $\operatorname{CSS}(\tilde C_X,\tilde C_Z)$, we use the GAP package QDistRnd \cite{Pryadko2022}.

\subsection{ Quantum Tanner code with cyclic local code}\label{section:Quantum Tanner code with cyclic local code}

\begin{figure*}[t]
  \centering
  \includegraphics[scale=1]{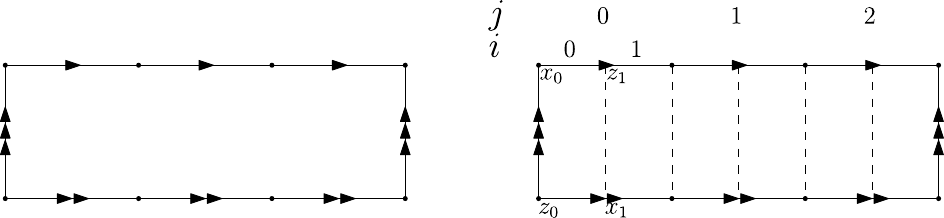}
  \caption{Left: space homotopy equivalent to the presentation complex of $\operatorname{L}(3)$. Right: its subdivision into a square complex and the indexing of the faces by $(i,j)\in [2]\times [\ell]$.}
  \label{fig:Quantum_Tanner_L}
\end{figure*}

Our first example of quantum Tanner code is obtained by considering a space that is homotopy equivalent to the presentation complex of the group defined by \[\operatorname{L}(\ell) :=  \langle a,b | a^\ell b^{-\ell}\rangle,\]
combined with a set of local tensor-product codes of length $2\ell$, and their dual. \par

\paragraph{Square complex.} A space with $\mathrm{L}(\ell)$ for fundamental group can be constructed on the model of Figure~\ref{fig:Quantum_Tanner_L}, corresponding to the case $\ell=3$. For higher values of $\ell$, we have $\ell$ identifications on the top and $\ell$ on the bottom. We build a square complex by subdividing each horizontal edge into two edges. We also connect the new vertices by additional vertical edges, which may be multi-edges, hence subdividing the unique face. This forms the bipartite square complex $\mathcal S_\ell=(V=V_X\sqcup V_Z,E,F)$ which has $|V|=4$, $|E|=2\ell+4$, and $|F|=2\ell$. The subdivision for $\ell=3$ is also shown in Figure~\ref{fig:Quantum_Tanner_L}, and the two diagonal graphs $\mathcal G_X^\Box,\:\mathcal G_Z^\Box $ are depicted in Figure~\ref{fig:diagonal graph L(l)}. Notice that the 1-skeleton of this complex is a bipartite multi-graph, which is adequate to define a quantum Tanner code.\par
\begin{figure}[t]
    \centering
    \includegraphics[width=0.5\linewidth]{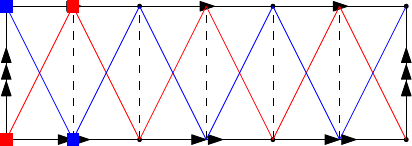}
    \caption{Diagonal graphs of the complex $\mathcal S_\ell$ corresponding to the group $\mathrm L(\ell)$ for $\ell=3$.}
    \label{fig:diagonal graph L(l)}
\end{figure}

We first identify the faces of this complex by the product set $[2]\times [\ell]$, where the index $(i,j)$ corresponds to the face in the $(i+1)\cdot (j+1)$-th position, when counting from left to right in Figure~\ref{fig:Quantum_Tanner_L}. We can see that each vertex is incident to all the $2\ell$ faces. The associated $2\ell$ edges of the diagonal graphs $\mathcal G_X^\Box $ and $\mathcal G_Z^\Box$ inherit the indexing of the faces of $\mathcal S_\ell$.\par

\paragraph{Local codes.} To construct local codes, we first consider $\mathscr{C}_2(1+X)$, the repetition code of length $2$, and $\mathscr{C}_\ell(g)$, the cyclic code of length $\ell$ generated by $g$. We define the tensor code $ \Pi_X, \Pi_Z\subseteq \mathbb F_2^{2}\otimes\mathbb F_2^{\ell}$, where 
\begin{align*}
    \Pi_X&=\mathscr{C}_2(1+X)\otimes\mathscr{C}_\ell(g), \\
    \Pi_Z&=\mathscr{C}_2(1+X)\otimes\mathscr{C}_\ell(g)^\perp,
\end{align*}
and $\mathbb F_2^{2}\otimes\mathbb F_2^{\ell}$ is endowed with the canonical product basis indexed by $[2]\times [\ell]$, as described in Section~\ref{section:local codes}. We will subsequently justify this particular choice of product codes.\par

\paragraph{Definition of $C_X$.} We start by defining a classical Tanner code associated to $\mathcal G_X^\Box$. The graph $\mathcal G_X^\Box=(V_X,F)$ has only two vertices, $V_X=\{x_{0},x_{1}\}$, as shown in Figure~\ref{fig:Quantum_Tanner_L}. The edge-neighborhood of these vertices, in the graph, is equal to their face-neighborhood in the complex $\mathcal S$, and $F(x_{0})=F(x_{1})=F$.\par
We notice a natural local product structure relative to the vertices of $V_X$. These are defined by the bijections
\begin{equation*}
   \phi_{x_0}:\begin{cases}
   \hfill F\hfill&\to [2]\times [\ell]\\
    (i,j)&\mapsto (i,i+j \operatorname{mod} \ell),\\
    \end{cases} \quad \text{and} \quad 
  \phi_{x_1}:\begin{cases}
      \hfill F\hfill&\to [2]\times [\ell]\\
       (i,j)&\mapsto (i,j).
   \end{cases} 
\end{equation*}
Each local code associated to these vertices is set to be isomorphic to the dual of the tensor code $C_{x_t}\cong \Pi_X^\perp$, $t=0,1$. For $t=0$, the isomorphism is induced by extending $\phi_{x_0}$ by linearity into $\phi_{x_0}:\mathbb F_2 F \to \mathbb F_2^{2}\otimes\mathbb F_2^{\ell}$ sending faces to the canonical basis vectors, where our choices of indexing for basis elements in the case of product codes and cyclic codes are described in Section~\ref{section:local codes}. Similarly, we obtain the bijection $\phi_{x_1}:\mathbb F_2 F \to \mathbb F_2^{2}\otimes\mathbb F_2^{\ell}$. This completes the characterization of the local codes $C_{V_X}=(C_{x_i})_{i \in \{0,1\}}$ and of the classical Tanner code $C_X :=  T(\mathcal G_X^\Box,C_{V_X})$.\par

\paragraph{Definition of $C_Z$.} The classical Tanner code associated to $\mathcal G_Z^\Box$ is defined similarly. The graph $\mathcal G_Z^\Box=(V_Z,F)$ also has two vertices, $V_Z=\{z_{0},z_{1}\}$, which are incident to all the faces, i.e.,  $F(z_{0})=F(z_{1})=F$. Each local code associated with these vertices is set to be isomorphic to the dual of the same tensor code, $C_{z_t}\overset{\phi_{z_t}}{\cong }\Pi_Z^\perp$ where each isomorphism is induced by $\phi_{z_t}=\phi_{x_t}$, for $t=0,1$. This completes the characterization of $C_{V_Z}=(C_{z_i})_{ i\in \{0,1\}}$ and $C_Z :=  T(\mathcal G_Z^\Box,C_{V_Z})$.\par

\paragraph{Orthogonality condition.}  Let us now explain in words why the code $\operatorname{CSS}(C_X,C_Z)$ is well-defined. The first factor of each product code $\Pi_X,\Pi_Z$ is a repetition code to ensure that the checks associated with $x_0 $ and $x_1$ commute with the checks associated with the vertices $z_0$ and $z_1$, respectively.\par
Moreover, due to the identification of the left and right edges of the presentation complex, our local product structure is not strict. To remedy this problem, the linear codes need to have a cyclic symmetry in order for a check associated with $x_0$ to commute with one associated with $z_1$. More precisely, with our choice of indexing, a codeword $c_x=(1,1)\otimes c$, representing a $X$-check, and a codeword $c_z=(1,1)\otimes c'$, representing a $Z$-check, are orthogonal to each other if \[ \langle c,c'\rangle +\langle c,\rho(c')\rangle=0,\]
where the rotation map $\rho$ is defined in Section~\ref{section:local codes}. This can be written as $\langle c,c'+\rho(c')\rangle=0$, which is automatically satisfied if the second factor of $\Pi_X$ and $\Pi_Z$ are a cyclic code and its dual, respectively. Overall, these conditions ensure that $C_X^\perp\subseteq C_Z$.\par

\paragraph{Parity-check matrices.} Another way to describe $C_X$ and $C_Z$ is by the kernels of parity-check matrices $H_X$ and $H_Z$, as explained in Section~\ref{section Tanner codes and their lifts}. We denote $h$ the check polynomial of $\mathscr{C}_\ell(g)$. In a certain basis, the parity-check matrix can be written
\begin{align*}
     H_X=\left[\begin{array}{c|c} 
 \mathbb G( Xg(X))  &  \mathbb G (g(X))\\ \hline
 \mathbb G (g(X)) &  \mathbb G( g(X))
\end{array}\right], \quad \quad
H_Z=\left[\begin{array}{c|c} 
 \mathbb G( \bar h(X))  &  \mathbb G (\bar h(X))\\ \hline
 \mathbb G (X\bar h(X)) &  \mathbb G( \bar h (X))
\end{array}\right].
\end{align*}
   
It is straightforward to verify that $H_X\cdot H_Z^T=0$, using that $\mathbb G( Xg(X))$ is a cyclic permutation of $\mathbb G( g(X))$, and that $\mathbb G(g(X)) \mathbb G(\bar h(X))^T=0$.\par

\paragraph{Code lifting.} Now that we have described quantum Tanner codes constructed on the square complex $\mathcal S_\ell$, we can consider their lifts, as given in Section~\ref{section lift of quantum Tanner codes}. Our first objective is to prove that our construction behaves well under lifting.

\begin{proposition}
    Let $p:\tilde{\mathcal S}_\ell\to \mathcal S_\ell$ be a covering map and $\tilde C_X=T(\tilde{\mathcal G}_X^\Box ,C_{\tilde V_X})$, $\tilde C_Z=T(\tilde{\mathcal G}_Z^\Box ,C_{\tilde V_Z})$, the lifted classical Tanner codes associated to $p$. Then $\tilde C_X^\perp\subseteq \tilde C_Z$.
\end{proposition}
\begin{proof}
Our objective is to determine the intersection of the face neighborhood $F(\tilde x_s)\cap F(\tilde z_t)$, for any pair of vertices $\tilde x_s\in p^{-1}(x_s)$ and $\tilde z_t\in p^{-1}(z_t)$, $s,t\in \{0,1\}$, and show that the support intersection of any $X$- and $Z$-checks associated with these vertices is of even size. By symmetry of the code structure, it is sufficient to verify this condition for pairs $(\tilde x_0,\tilde z_0)$ and $(\tilde x_0,\tilde z_1)$ in $\tilde S_\ell$.\par
It is clear that any adjacent vertices $\tilde x_0,\tilde z_0$ share at least 2 faces, which are attached to an edge $\tilde e=\{\tilde x_0,\tilde z_0\}$, drawn vertically on Figure~\ref{fig:Quantum_Tanner_L}. Since an $X$-check is a codeword of the product code $\Pi_X$, if it acts on the qubit associated to one of the two faces, then it must also act on the other, and similarly for a $Z$-check. Therefore, the support of an $X$- and $Z$-check associated with the vertices $\tilde x_0 $ and $\tilde z_0$, respectively, must intersect on an even number of qubits.\par
A pair of adjacent vertices of the form $\tilde x_0,\tilde z_1$ are connected by an edge incident to $\ell$ faces in $\tilde S_\ell$. The restriction of the support vector of an $X$-check associated with $\tilde x_0$ on these faces is a codeword of $\mathscr{C}_\ell(g)$, while that of a $Z$-check associated $\tilde z_1$ is a codeword of $\mathscr{C}_\ell(g)^\perp$. These vectors being orthogonal, their support intersection is even.
\end{proof}

\begin{table}[]
\centering 
\begin{tabular}{l|l|l|l|l|l}

$\ell$&  $W$&Lift index& $\operatorname{deck}(p)$& $[[n,k,d]]$   &$d^2/n$\\
\hline \hline 
 10& 4& 1& $1$&$[[20,2,2]]$ &0.2\\
 & & 20& $D_{20}$, $\mathbb Z_5 \rtimes \mathbb Z_4$&$[[400,2,20]]$ &1\\ \hline   
  14&     12 &1     &   $1$&                     $[[28,2,6]]$ &1.25\\
 & & 4& $\mathbb Z_2 \times \mathbb Z_2$, $\mathbb Z_4$&$[[112,2,12]]$ &1.25\\
 & & 7& $\mathbb Z_7$&$[[196,2,18]]$ &1.65\\
 & & 16& $\mathbb Z_4 \rtimes \mathbb Z_4$, $\mathbb Z_8 \rtimes \mathbb Z_2$, $Q_{16}$&$[[448,2,24]]$           &1.25\\ 
  &          &28&       $\mathbb Z_{28}$&                     $[[784,2,36]]$ &1.65\\

\end{tabular}
\caption{Parameters of selected lifted quantum Tanner codes built from the space of Figure~\ref{fig:Quantum_Tanner_L}, whose homotopy group is isomorphic to $\operatorname{L}(\ell)$. Different groups in a single entry of the 4th column correspond to different lifted codes with the same parameters. Here, $W$ denotes the maximum row or column-weight of the parity-check matrices. The value of $d$ is an upper bound found with the GAP package QDistRnd \cite{Pryadko2022}, and in all these cases $d_X=d_Z$. The specific local codes involved in these constructions are described in Examples~\ref{example L1} and~\ref{example L2}. }\label{Table:lift L(l)}
\end{table}
To illustrate this construction, we instantiate the case $\ell=10$ and $\ell=14$ in the examples below, and we summarize the parameters of selected lifted codes in Table~\ref{Table:lift L(l)}.
 \begin{example}\label{example L1}
The instances for $\ell=10$ in Table~\ref{Table:lift L(l)} are obtained for a cyclic code $\mathscr{C}_\ell(g)$ with generating polynomial $g(X)=1+X^5$, and check polynomial $h(X)=g(X)$. Used as a local code, this produces Tanner codes with check-weight 4. As in the Tanner version of the toric code, they could be seen as rotated versions of certain topological codes.
\end{example}

\begin{example}\label{example L2}
    The instances for $\ell=14$ in Table~\ref{Table:lift L(l)}, are obtained by using the shortest example of non-trivial cyclic self-dual binary code \cite{sloane_cylic}. This has parameters $[14,7,4]$, and is generated by the polynomial 
\[g(X)=(X+1)(X^3 +X+1)^2=1+X+X^2+X^3+X^6+X^7.\]
The generator matrix associated to $g$ has row-weight $6$, but by elementary operations, we can turn it into a full-rank generator matrix with $6$ rows of weight $4$, and one row of weight $6$. This means that for the associated quantum Tanner code and its lifts, one-seventh of the checks are of weight $12$ and the others are of weight $8$. We emphasize that the last four lines of Table~\ref{Table:lift L(l)} represent LDPC codes with distance $d>\sqrt{n}$.
\end{example}

\begin{remark}\label{remark quantum Tanner lift beyond other lift}
In the base code, an $X$- and a $Z$-check overlap on all the qubits, while in the lifted code, their overlap must be reduced, otherwise the square complex would not remain connected. This shows that these codes, which do not satisfy the assumption of Proposition \ref{proposition quantum Tanner lifting simply connected closure}, constitute examples of lifts beyond the technique of \cite{GuemardLiftIEEE}, where lifting a code preserved the size of the overlap between any pair of $X$- and $Z$-checks.
\end{remark}

\subsection{Quantum Tanner code with double-circulant local code}\label{section:Quantum Tanner code with double-circulant local code}

\begin{figure*}[t]
  \centering
 \includegraphics[scale=1]{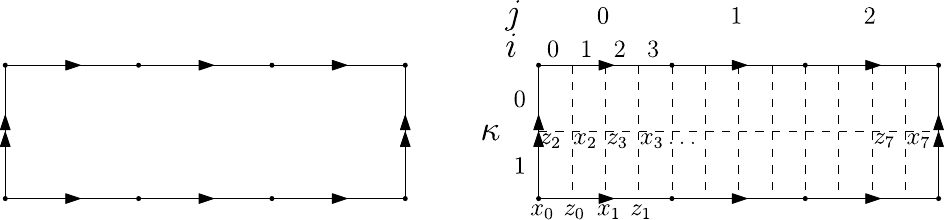}
  \caption{Left: presentation complex $M_\ell$ of $\operatorname{BS}(\ell,\ell)$ for $\ell=3$. Right: its subdivision into a bipartite square complex $\mathcal S_3$, and the indexing of the faces by $(i,j,\kappa)\in [4]\times \mathbb [\ell]\times \mathbb [2]$.}
  \label{fig:Quantum_Tanner_BS}
\end{figure*}

Our second construction is based on the presentation complex $M_\ell$ of the Baumslag-Solitar group defined as
\[\operatorname{BS}(\ell,\ell)=\langle a,b | ab^\ell a^{-1}b^{-\ell}\rangle,\]
where $\ell \in \mathbb N$, combined with a set of local tensor-product codes.

\paragraph{Square complex.} The presentation complex of $\mathrm{BS}(\ell)$ is constructed on the model of Figure~\ref{fig:Quantum_Tanner_L}, which is for $\ell=3$. For higher values of $\ell$, we identify $\ell$ segments along the top and $\ell$ along the bottom. The presentation complex always has two directed edges, $e_a$ and $e_b$, associated with the generators $a$ and $b$ of the group presentation, respectively. It is modified into a bipartite square complex by subdividing the unique face into two rows of $4\ell$ square faces. This process also subdivides the horizontal edge $e_b$ into four edges (indexed by $[4]$), and the vertical edge $e_a$ into two edges (indexed by $[2]$). This also adds new new horizontal and vertical edges. This forms the bipartite square complex $\mathcal S_\ell := (V=V_X\sqcup V_Z,E,F)$ having $|V|=4\ell+4$, $|E|=8\ell+4$, and $|F|=8\ell$. Each new vertex of the subdivided edge $e_b$ is incident to $4\ell$ faces, and each vertex along the middle cycle of edges is incident to 4 faces. The two diagonal graphs $\mathcal G_X^\Box$ and $\mathcal G_Z^\Box $ of $\mathcal S_\ell$ are depicted in Figure~\ref{fig:diagonal graph L(l)}\par

\begin{figure}[t]
    \centering
    \includegraphics[width=0.5\linewidth]{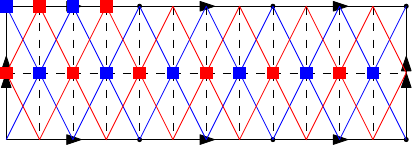}
    \caption{Diagonal graphs of the complex $\mathcal S_\ell$ corresponding to the group $\mathrm{BS}(\ell)$ for $\ell=3$.}
    \label{fig:enter-label}
\end{figure}

We identify the faces of this complex with elements of the product set $[4]\times [\ell]\times [2] $, where $(i,j,\kappa)$ represents the face in the upper row if $\kappa=0$ and in the upper row if $\kappa=1$, placed in the $(i+1)\cdot (j+1)$-th position when counting from left to right in Figure~\ref{fig:Quantum_Tanner_BS}. The face $(i,j,\kappa)$ is therefore incident to the edge indexed $i$ in the subdivision of $e_b$. The $8\ell$ edges of the diagonal graphs $\mathcal G_X^\Box $ and $\mathcal G_Z^\Box$ inherit the indexing of the faces of $\mathcal S_\ell$.\par

\paragraph{Local codes.} To construct the local codes, we consider $\mathscr{C}_2(1+X)$, the repetition code of length $2$, and the double-circulant code $\mathscr{D}_{2\ell}(f)$, with polynomial $f\in R_\ell$. We define the tensor codes $\Pi_X, \Pi_Z\subseteq \mathbb F_2^{2}\otimes\mathbb F_2^{2\ell}$, where 
\begin{align*}
    \Pi_X&=\mathscr{C}_2(1+X)\otimes\mathscr{D}_{2\ell}(f) \\
    \Pi_Z&=\mathscr{C}_2(1+X)\otimes\mathscr{D}_{2\ell}(f)^\perp.
\end{align*}
Recall that $\mathbb F_2^{2}\otimes\mathbb F_2^{2\ell}$ is endowed with the product basis indexed by the set $[2]\times [2\ell]$ as described in Section~\ref{section:local codes}. We will subsequently justify why this particular choice of product codes ensures orthogonality of $X$- and $Z$-checks.\par

\paragraph{Definition of $C_X$.} We first define a classical Tanner code on $\mathcal G_X^\Box=(V_X,F)$, where $V_X=(x_i)_{i\in [2\ell+2]}$. This graph has two-vertices of degree $4\ell$, labeled $x_0$ and $x_1$, which result from subdividing the edge $e_b$. Their edge neighborhood in $\mathcal G_X^\Box$ is equal to their face neighborhood in $\mathcal S_\ell$, and they satisfy $F(x_{0})\cap F(x_{1})=\emptyset$, $F(x_{0})\cup F(x_{1})=F$.\par
We notice a natural local product structure relative to the vertices of $V_X$. These are defined by the bijections
\begin{equation*}
   \phi_{x_0}:\begin{cases}
   \hfill F(x_0)\hfill&\to [2]\times [2\ell]\\
     (0,j,\kappa)&\mapsto (0,\kappa\ell +j)\\
     (3,j,\kappa)&\mapsto (1,\kappa\ell + (j+1\operatorname{mod}\ell))
    \end{cases} \quad \text{and} \quad 
  \phi_{x_1}:\begin{cases}
      \hfill F\hfill&\to [2]\times [\ell]\\
      (i,j,\kappa)&\mapsto (i\operatorname{mod}2,\kappa\ell +j).
   \end{cases} 
\end{equation*}
Each local code associated with these vertices is set to be isomorphic to the dual of the tensor code $C_{x_t}\cong \Pi_X^\perp$, where the isomorphism is obtained by extending $\phi_{x_t}$ by linearity into $\phi_{x_t}:\mathbb F_2 F(x_t) \to \mathbb F_2^{2}\otimes\mathbb F_2^{2\ell}$ for $t=0,1$,  sending each face to a canonical basis vectors.\par

The graph $\mathcal G_X^\Box$ also has $2\ell$ vertices of degree $4$, labeled $x_t$, $t=2,\dots, 2\ell+1$. They belong to the central horizontal cycle of edges, as shown in Figure~\ref{fig:Quantum_Tanner_BS}. For all of them, we set $C_{x_t}$ to be the parity code of length 4, i.e. the single parity check assigned to vertex $x_t$ has full support over its face neighborhood $F(x_t)$ in $\mathcal S_\ell$. This completes the characterization of $C_{V_X}=(C_{x_i})_{i\in [2\ell+1]}$ of the classical Tanner code $C_X :=  T(\mathcal G_X^\Box,C_{V_X})$.\par

\paragraph{Definition of $C_Z$.}The classical Tanner code associated to  $\mathcal G_Z^\Box$ is defined similarly. The graph $\mathcal G_Z^\Box=(V_Z,F)$, where $V_Z=\{z_i,\:i\in [2\ell+2]\}$, has two vertices of degree $4\ell$ labeled $z_0$ and $z_1$. They satisfy $F(z_{0})\cap F(z_{1})=\emptyset$, $F(z_{0})\cup F(z_{1})=F$.\par
We notice again a natural local product structure relative to the vertices of $V_Z$, that is defined by the bijections
\begin{equation*}
    \phi_{z_t}:\begin{cases}
        F(z_t)&\to [2]\times [2\ell]\\
          (i,j,\kappa)&\mapsto (i\operatorname{mod}2,\kappa\ell +j),
    \end{cases}
\end{equation*}
for $t=0,1$. For each of these vertices, we define a local code $C_{z_t}\cong\Pi_Z^\perp$, where this isomorphism is obtained by extending $\phi_{z_t}$ into a map $\phi_{z_t}:\mathbb F_2 F(z_t) \to\mathbb F_2^{2}\otimes\mathbb F_2^{2\ell}$. The graph $\mathcal G_Z^\Box$ has also $2\ell$ vertices of degree $4$, labeled $z_t$, $t=2,\dots ,2\ell+1$. They belong to the central horizontal cycle of edges, as depicted in Figure~\ref{fig:Quantum_Tanner_BS}. For them, we set $C_{z_t}$ to be the parity code of length 4. This completes the characterization of $C_{V_Z}=(C_{z_i})_{i\in [2\ell+1]}$ and $C_Z :=  T(\mathcal G_Z^\Box,C_{V_Z})$.\par

\paragraph{Orthogonality condition.} We next explain why the code $\operatorname{CSS}(C_X,C_Z)$ is well-defined. To begin with, the first factor of the product codes $\Pi_X,\Pi_Z$ is a repetition code to ensure that the checks associated with $x_t,z_t,t=0,1$ commute with the middle checks, those for $t=2,\dots ,2\ell+1$.\par
Next, due to the identification of the left and right edges of the presentation complex, our local product structure is not strict. To resolve this issue, the second linear codes must have a simultaneous cyclic symmetry for $\kappa=0,1$ so that a check associated with $x_0$ commutes with one associated with $z_1$. To be precise, following our choice of indexing for the local codes, the condition for a codeword $c_x=(1,1)\otimes c$, representing an $X$-check, and a codeword $c_{z}=(1,1)\otimes c'$, representing a $Z$-check, to be mutually orthogonal is
\[ \langle c,c'\rangle=0\:\text{and}\:\langle \sigma(c),c'\rangle=0,\]
where the map $\sigma$ is defined in Section~\ref{section:local codes}. This is automatically satisfied if the second factor of $\Pi_X$ and $\Pi_Z$ are a double-circulant code and its dual, respectively. The two conditions above thus ensure $C_X^\perp\subseteq C_Z$.\par

\paragraph{Code lifting.} Now that we have described quantum Tanner codes constructed over the square complex $\mathcal S_\ell$, we can consider their lifts, as in Section~\ref{section lift of quantum Tanner codes}. Our first objective is to prove that our construction behaves well under the operation of code lifting.

\begin{proposition}
Let $p:\tilde{\mathcal S}_\ell\to \mathcal S_\ell$ be a covering map and $\tilde C_X=T(\tilde{\mathcal G}_X^\Box ,C_{\tilde V_X})$, $\tilde C_Z=T(\tilde{\mathcal G}_Z^\Box ,C_{\tilde V_Z})$ be the lifted classical Tanner codes associated to $p$. Then $\tilde C_X^\perp\subseteq \tilde C_Z$.
\end{proposition}

\begin{proof}
Our objective is to determine the intersection of the face neighborhood $F(\tilde x_s)\cap F(\tilde z_t)$, for any pair of vertices $\tilde x_s\in p^{-1}(x_s)$ and $\tilde z_t\in p^{-1}(z_t)$, $s,t\in \{0\dots 2\ell+1\}$, and show that the support intersection of any $X$- and $Z$-checks associated with these vertices is of even size. By symmetry of the code, it is sufficient to check for pairs $(\tilde x_0,\tilde z_0)$ and $(\tilde x_0,\tilde z_2)$ in $\tilde S_\ell$.\par
It is clear that any adjacent vertices $\tilde x_0,\tilde z_2$ are incident to either two or four common faces. It is possible to partition this set of faces into pairs of adjacent faces and sharing an edge $\tilde e$ whose image by $p$ is an edge $\{x_0,z_2\}$ represented vertically in Figure~\ref{fig:Quantum_Tanner_BS}. Since an $X$-check is a codeword of the product code $\Pi_X=\mathscr{C}_2(1+X)\otimes\mathscr{D}_{2\ell}(f)$, if it acts on the qubit associated with one of the two faces of a given pair, then it must also act on the other. Moreover, a $Z$-check has full support on the qubits associated to the face neighborhood of $\tilde z_2$ . Therefore, the support of an $X$-and a $Z$-check associated with the vertices $\tilde x_0 $ and $\tilde z_0$, respectively, must intersect on an even number of qubits.\par
A pair of adjacent vertices of the form $\tilde x_0,\tilde z_1$ are connected by an edge incident to $2\ell$ faces in $\tilde S_\ell$. The restriction of the support vector of an $X$-check associated with $\tilde x_0$ is a codeword of $\mathscr{D}_{2\ell}(f)$ on these $2\ell$ faces, while that of a $Z$-check associated with $\tilde z_1$ is a codeword of $\mathscr{D}_{2\ell}(f)^\perp$. These vectors being orthogonal, their support intersection is even.
\end{proof}

\begin{table}[]
\centering 
\begin{tabular}{l|l|l|l|l|l}

$\ell$&  $W$&lift index& $\operatorname{deck}(p)$& $[[n,k,d]]$ &$d^2/n$\\\hline \hline 
 3& 6& 1& $1$&$[[24,0,\infty]]$&$\infty$\\
 & & 12& $\mathbb Z_{12}$, $D_{12}$, $\mathbb Z_3 \rtimes \mathbb Z_4$&$[[288,4, 6]]$ &0.125\\
 & & 24&  $\mathbb Z_2\times (\mathbb Z_3 \rtimes \mathbb Z_4)$,&$[[576,4,24]]$ &1\\
 & & & $\mathbb Z_3 \rtimes \mathbb Z_8$, $\mathbb Z_4\times S_3$& &\\ \hline  
  4&     8 &1     &   $\{e\}$&                     $[[32,2,4]]$          &0.5\\  
  &          &3&       $\mathbb Z_3$&                     $[[96,2,12]]$           &1.5\\   
  &          &5&       $\mathbb Z_5$&                     $[[160,2,16]]$          &1.6\\ 

\end{tabular}
\caption{Parameters of selected lifted quantum Tanner codes built over the presentation complex of $\operatorname{BS}(\ell,\ell)$. Different groups in an entry of the 4th column correspond to different lifted codes with the same parameters. Here $W$ denotes the maximum row or column-weight of the parity-check matrices. The value of $d$ is an upper bound found with the GAP package QDistRnd \cite{Pryadko2022}, and in all these cases, $d_X=d_Z$. The specific local codes involved in these constructions are described in Examples~\ref{example BS1} and~\ref{example BS2}. }\label{Table:lift BS(l)}
\end{table}

\begin{figure}[t]
    \centering
    \includegraphics[width=0.5\linewidth]{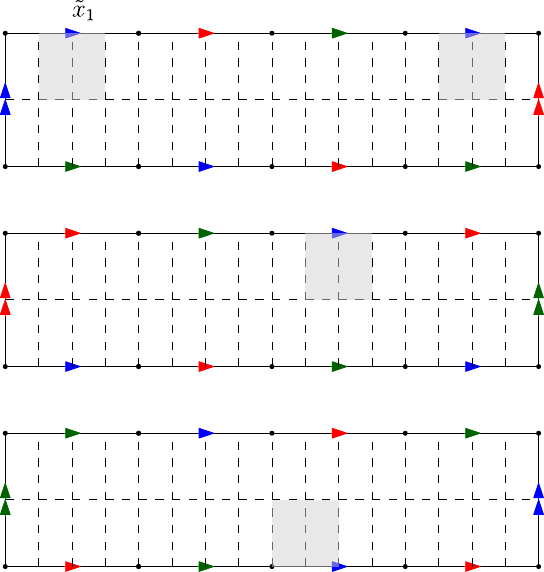}
    \caption{Covering space of the square complex $\mathcal S_4$ associated to the group $\mathbb Z_3$ in Table \ref{Table:lift BS(l)} and Example \ref{example BS2}. Directed edges with an arrow of the same color and type are identified. The support of an $X$-stabilizer associated with the local code defined in $\tilde x_1\in p^{-1}(x_1)$ is represented by the eight gray faces. The other checks of weight 8 are obtained by cyclic shifts of this pattern. Each vertex of the middle cycle supports a check of weight four (on its entire face neighborhood). This leads to the code with parameters $[[96,2,12]]$. This figure directly provides a planar layout of the physical qubits.}
    \label{fig:quantum_chip}
\end{figure}

We examine two specific instances and their lifts in the examples below, for $\ell=3$ and $\ell=4$. We list the parameters of selected lifted codes in Table~\ref{Table:lift BS(l)}.\par

\begin{example}\label{example BS1}
The smallest non-trivial quantum Tanner code arises for $\ell=3$. In Table~\ref{Table:lift BS(l)}, the double-circulant code used in the tensor product is either the $[6,3,3]$ code $\mathscr{D}_{6}(f)$ with $f(X)=X+X^2$, or its dual. In this case, we notice that the quantum Tanner code has zero dimension. However, by lifting it, we find codes of check weight $6$ with positive dimension.
\end{example}

The following example, corresponding to the last three rows of Table~\ref{Table:lift BS(l)} constitutes our best lifted codes.

\begin{example}\label{example BS2}
  The case $\ell=4$ in Table~\ref{Table:lift BS(l)} is obtained by using the shortest example of non-trivial double circulant self-dual binary code. This is the $[8,4,4]$ extended Hamming code $\mathscr{D}_8(f)$ with $f(X)=X+X^2+X^3$. The generator matrix of $\mathscr{D}_8(f)$ has row-weight 4, meaning that each of the classical Tanner codes $C_X$ and $C_Z$ has 16 parity checks, 8 of which being of weight $8$. The other 8 are defined on vertices indexed by $t=2,\dots, 9$ and are checks of parity codes of weight $4$. The weights of the parity-checks are found in the same proportions in the lifted codes. We emphasize that in the last two lines of Table~\ref{Table:lift BS(l)}, the codes are obtained by cyclic lifts, and their distances satisfy $d>\sqrt{n}$.\par
  In Figure \ref{fig:quantum_chip}, we illustrate the covering of $\mathcal S_4$ associated with the $[[96,2,12]]$ code. Because of the symmetry of the non-local checks, all obtained by cyclic shifts of the one represented in the figure, this may provide a biplanar layout of the physical qubits suitable for experimental implementation. We leave a rigorous analysis for future work.
\end{example}
 \begin{remark}
The codes introduced in this section do not satisfy the assumption of Proposition~\ref{proposition quantum Tanner lifting simply connected closure} and also constitute examples beyond the lifting technique of \cite{GuemardLiftIEEE}, similarly to those described in Remark~\ref{remark quantum Tanner lift beyond other lift}.
\end{remark}

Propositions~\ref{proposition transfer homomorphism} and~\ref{proposition parameter bound} can be applied to the last two codes of Table~\ref{Table:lift BS(l)}, corresponding to Example~\ref{example BS2}. In those cases, the codes with parameters $[[96,2,12]]$ and $[[160,2,16]]$ are index-$t$ lifts with $t=3$ and $t=5$, respectively, and the group of deck transformations is $\Gamma=\mathbb Z_t$. Using Proposition~\ref{proposition transfer homomorphism}, we obtain a lower bound on the parameters of these codes, which is consistent with $k=2$. Moreover, since in those cases lifting preserves the dimension of the base code, $\pi^*$ and $\tau_*$ are isomorphisms. Therefore, the logical operators, which correspond to classes in $H^i(\tilde C)$ and $H_i(\tilde C)$, respectively, are invariant under cyclic permutations of $\Gamma$. Finally, the $[[96,2,12]]$ code saturates the distance upper bound given by Proposition~\ref{proposition parameter bound}.

\section*{Acknowledgement}
The first author would like to thank Benjamin Audoux and Anthony Leverrier for valuable discussions throughout this work.  We acknowledge the Plan France 2030 through the project NISQ2LSQ ANR-22-PETQ-0006.
\newpage
 
\bibliographystyle{alpha}
%\bibliography{main}
\newcommand{\etalchar}[1]{$^{#1}$}

\newpage

\appendix

\section{Connected covering of a square-complex}

In this section, we describe a procedure to explicitly construct a connected covering of a square-complex, which was introduced in \cite{GuemardLiftIEEE}. It is based on the voltage assignment technique of \cite{GROSS1977273,Panteleev2021} and adapted from \cite{Lyndon2001,HatcherTopo} to generate connected covering maps.\par

We consider a connected square-complex $\mathcal S=(V,E,F)$ which has finitely many cells. Recall that the graph $\mathcal G=(V,E)$, can be seen as the 1-skeleton of $\mathcal{S}$. The fundamental group of a square-complex is defined in Section \ref{Graphs, square-complexes and covering maps}.\par
From the Galois correspondence, Theorem \ref{Theorem Galois correspond}, connected coverings of $\mathcal S$ are in one-to-one correspondence with subgroups of $\pi_1(\mathcal{S},v)$, where $v$ is an arbitrary choice of basepoint that we later omit. Let $H$ be a subgroup of $\pi_1(\mathcal S)$. There exists a unique covering map, 
\begin{equation}\label{equation covering}
p:\mathcal{S}_H \rightarrow \mathcal{S},
\end{equation}
associated to $H$, and it has index $r=[\pi_1(\mathcal S):H]$. To construct the covering above we will construct one for $\mathcal G$. For that, we have to consider the following subgroups and homomorphisms,
\[
\begin{tikzcd}[column sep=0cm, row sep=0.6cm]
  &  \pi_1(\mathcal S) 
  & \geq 
  &   H \arrow{d}{\phi^{-1}}
\\
  & \pi_1(\mathcal G) \arrow{u}{\phi} 
  & \geq
  & H^1
\end{tikzcd},
\]
where $\phi$ is the natural homomorphism taking the quotient of $\pi_1(\mathcal G)$ by the normal closure of the subgroup generated by relations coming from 4-cycles in $\mathcal{S}$. Moreover, $H^1$ is the preimage of $H$ in $\pi_1(\mathcal G)$ by this homomorphism. We now denote $\Gamma:=\pi_1(\mathcal{K} )/H$ and $\Gamma^1:=\pi_1(\mathcal G)/H^1$. When $H$ is normal, $\Gamma $and $\Gamma^1$ are isomorphic, otherwise this is a bijection between the cosets (here, all cosets are taken on the right). \par
To construct the covering map of Equation \eqref{equation covering}, it is shown in \cite{Lyndon2001} that we only need to construct the following covering,
\begin{equation}\label{Equation covering of 1-skeleton}
p_{1}:\mathcal G_{H^1}\rightarrow \mathcal G,
\end{equation}
which is unique by Theorem \ref{Theorem Galois correspond} .\par
We explain the steps to create the components of Equation \eqref{Equation covering of 1-skeleton}, and we detail a combinatorial technique specifically designed to produce connected coverings. The graph $\mathcal G_{H^1}$ has a set of vertices $V\times \Gamma^1$, and set of edges in bijection with $E\times \Gamma^1$. We now describe how to connect the edges. Since $\mathcal G$ is a connected graph, it admits a maximal spanning tree $T$ with basepoint $v$, and by definition a unique reduced path from vertex to vertex. We denote $E^*$ the set of oriented edges (its cardinality is twice the one of $E$).\par
We define a map called \textit{voltage assignment} on the set of oriented edges,
\begin{equation}\label{Equation modified voltage assignment}
    \nu^T: E^* \to \pi_1(\mathcal G),
\end{equation}
mapping an edge $e=[v_1,v_2]$ to an element of $\pi_1(\mathcal G)$ in the following manner. Let a path $\alpha$ from vertex $a$ to $b$ be denoted as a sequence of oriented edges, for example $\alpha=[a,u_1].[u_1,u_2]\dots [u_j,b]$. Moreover, let $((a,b))$ denote the path from $a$ to $b$ in the tree $T$. Then the equivalence class of circuits obtained by adding $e$ to $T$ is $\nu^T(e):=[((v,v_1)).e.((v_2,v))]$, where $[\alpha]$ is the standard notation for the homotopy class of a circuit $\alpha$ based at $v$.\par
We define a connected lift of the graph $\mathcal G$ associated to a certain subgroup of its fundamental group, using the notation above.

\begin{definition}[Connected lift of a graph]\label{definition modified lift of graph}
The lift of $\mathcal G$ associated to the subgroup $H^1\leq \pi_1(\mathcal G)$ is the graph $\mathcal G_{H^1}$ with set of vertices $V\times \Gamma^1$ and set of edges in bijection with $E\times \Gamma^1$, so that a vertex or edge is written $(c,gH^1)$ with $c\in \mathcal G$ and $g\in \pi_1(\mathcal G)$. An oriented edge $\tilde e=(e,gH^1)$ in the graph $\mathcal G_{H^1}$, with $e=[u,v]$,  connects $(u,gH^1)$ and $(v, g \nu^T(e)H^1)$, where the multiplication by $\nu^T(e)$ is on the right. 
\end{definition}
When $H^1$ is a normal subgroup of $\pi_1(\mathcal G)$, the resulting covering is Galois.\par
In general, a graph admits more than one spanning tree. Following this procedure with another choice of spanning tree will modify the parametrization of the vertices and edges, but the graphs will be isomorphic, by the Galois correspondence. If $H^1$ is chosen as an arbitrary subgroup of $\pi_1(\mathcal G)$, it can be shown \cite{GuemardLiftIEEE, Lyndon2001} that Definition \ref{definition modified lift of graph} can generate all connected regular and non-regular covers of $\mathcal G$. We define $p_{1}$ as the map projecting any cell of $\mathcal G_{H_1}$, vertex or edge, to its first coordinate. This is a valid covering map, equivalent to what we have done in Section \ref{section covering maps}.\par

Here $H^1$ is not arbitrary; it is the preimage of $H$ by $\phi$. It is also possible to prove that $p_{1}$ induces an isomorphism $p_{1\#}:\pi_1(\mathcal G_{H^1},\tilde v)\to H^1$, for any choice of basepoint $\tilde v$ in the preimage of $v$, see \cite{Lyndon2001}, Chapter III. 3. Moreover, for a 4-cycle $f\in F$, we can show that $p_1^{-1}(f)$ is a set of $r$ 4-cycles, and hence $p_1$ can be extended into a covering of the square-complex $\tilde{\mathcal S}:=(p_1^{-1}(V),p_1^{-1}(E), p_1^{-1}(F))$. It is finally proven in \cite{Lyndon2001} that $p_1$ induces an isomorphism $p_{1\#}:\pi_1(\mathcal S_H,v')\to H$ so that $p_H=p_1$ and $\tilde{\mathcal S}=\mathcal S_H$. \par

\end{document}